\documentclass[a4paper,11pt,reqno]{amsart}

\usepackage[margin=1in,includehead,includefoot]{geometry}
\usepackage[utf8]{inputenc}
\usepackage[T1]{fontenc}
\usepackage{lmodern,microtype}
\usepackage{hyperref}
\usepackage{booktabs}
\usepackage{mathtools,amssymb,amsthm,amsmath}
\usepackage[foot]{amsaddr}
\usepackage{graphicx}
\usepackage[table]{xcolor}
\usepackage{wrapfig}
\usepackage[margin=5mm]{caption}
\usepackage{enumitem}
\usepackage[textsize=footnotesize,disable]{todonotes}
\usepackage{mdframed}
\usepackage{xpatch}
\usepackage{multirow}
\usepackage{hhline}
\usepackage{dsfont}
\usepackage{mycommands}

\graphicspath{{./graphics/}}

\makeatletter
\let\old@setaddresses\@setaddresses
\def\@setaddresses{\bigskip{\parindent 0pt\let\scshape\relax\let\ttfamily\relax\old@setaddresses}}
\makeatother

\newcommand{\torsten}[1]{\todo[inline,color=blue!40]{\textbf{Torsten:} #1}}

\newtheorem{theorem}{Theorem}
\newtheorem{corollary}[theorem]{Corollary}
\newtheorem{lemma}[theorem]{Lemma}
\theoremstyle{remark}
\newtheorem{remark}[theorem]{Remark}

\hypersetup{
  pdftitle={Traversing combinatorial 0/1-polytopes via optimization},
  pdfauthor={Arturo Merino and Torsten M\"utze}
}

\begin{document}

\title{Traversing combinatorial 0/1-polytopes via optimization}

\author{Arturo Merino}
\address[Arturo Merino]{Department of Mathematics, TU Berlin, Germany}
\email{merino@math.tu-berlin.de}

\author{Torsten M\"utze}
\address[Torsten M\"utze]{Department of Computer Science, University of Warwick, United Kingdom \& Department of Theoretical Computer Science and Mathematical Logic, Charles University, Prague, Czech Republic}
\email{torsten.mutze@warwick.ac.uk}

\thanks{An extended abstract of this paper appeared as~\cite{MR4720318} in the Proceedings of FOCS~2023.
Arturo Merino was supported by ANID Becas Chile 2019-72200522.
Torsten M\"utze was supported by Czech Science Foundation grant GA~22-15272S.
Both authors participated in the workshop `Combinatorics, Algorithms and Geometry' in March 2024, which was funded by German Science Foundation grant~522790373.}

\begin{abstract}
In this paper, we present a new framework that exploits combinatorial optimization for efficiently generating a large variety of combinatorial objects based on graphs, matroids, posets and polytopes.
Our method is based on a simple and versatile algorithm for computing a Hamilton path on the skeleton of a 0/1-polytope $\conv(X)$, where $X\seq \{0,1\}^n$.
The algorithm uses as a black box any algorithm that solves a variant of the classical linear optimization problem~$\min\{w\cdot x\mid x\in X\}$, and the resulting delay, i.e., the running time per visited vertex on the Hamilton path, is larger than the running time of the optimization algorithm only by a factor of $\log n$.
When $X$ encodes a particular class of combinatorial objects, then traversing the skeleton of the polytope~$\conv(X)$ along a Hamilton path corresponds to listing the combinatorial objects by local change operations, i.e., we obtain Gray code listings.

As concrete results of our general framework, we obtain efficient algorithms for generating all ($c$-optimal) bases and independent sets in a matroid; ($c$-optimal) spanning trees, forests, matchings, maximum matchings, and $c$-optimal matchings in a graph; vertex covers, minimum vertex covers, $c$-optimal vertex covers, stable sets, maximum stable sets and $c$-optimal stable sets in a bipartite graph; as well as antichains, maximum antichains, $c$-optimal antichains, and $c$-optimal ideals of a poset.
Specifically, the delay and space required by these algorithms are polynomial in the size of the matroid ground set, graph, or poset, respectively.
Furthermore, all of these listings correspond to Hamilton paths on the corresponding combinatorial polytopes, namely the base polytope, matching polytope, vertex cover polytope, stable set polytope, chain polytope and order polytope, respectively.

As another corollary from our framework, we obtain an $\cO(t_{\upright{LP}} \log n)$ delay algorithm for the vertex enumeration problem on 0/1-polytopes $\{x\in\mathbb{R}^n\mid Ax\leq b\}$, where $A\in \mathbb{R}^{m\times n}$ and~$b\in\mathbb{R}^m$, and $t_{\upright{LP}}$ is the time needed to solve the linear program $\min\{w\cdot x\mid Ax\leq b\}$. 
This improves upon the 25-year old $\cO(t_{\upright{LP}}\,n)$ delay algorithm due to Bussieck and L\"ubbecke.
\end{abstract}

\maketitle

\section{Introduction}

In mathematics and computer science, we frequently encounter different classes of combinatorial objects, for example spanning trees, matchings or vertex covers of a graph, independent sets or bases of a matroid, ideals or antichains of a poset, etc.
Given a class~$X$ of objects, in \defi{combinatorial optimization} we are interested in finding the best object from~$X$ w.r.t.\ some objective function~$f$, i.e., we aim to compute a minimizer of~$f(x)$ over all $x\in X$.
Classical examples for such problems on graphs are computing a minimum weight spanning tree, a maximum weight matching, or a minimum size vertex cover.
Motivated by countless practical instances of such problems, the field of combinatorial optimization has matured over decades into a huge body of work (see e.g.~\cite{MR1490579,MR1956924,MR1956925,MR1956926,MR3753583}), which combines powerful algorithmic, combinatorial and polyhedral methods.
Important techniques include dynamic programming, linear and integer programming, network flows, branch-and-bound and branch-and-cut methods, approximation algorithms, parametrized algorithms, etc.

Another fundamental algorithmic task apart from combinatorial optimization is \defi{combinatorial generation}, covered in depth in Knuth's book~\cite{MR3444818}.
Given a class~$X$ of objects, the task here is to exhaustively list each object from~$X$ exactly once.
The running time of a generation algorithm is typically measured by its \defi{delay}, i.e., by the time spent between generating two consecutive objects from~$X$.
Sometimes it is reasonable to relax this worst-case measure, and to consider the \defi{amortized delay}, i.e., the total time spent to generate~$X$, divided by the cardinality of~$X$.
Algorithms that achieve delay~$\cO(1)$ are the holy grail, and they are sometimes called \defi{loopless}.
Compared to combinatorial optimization, the area of combinatorial generation is much less matured.
Nonetheless, some general techniques are available, such as the flashlight search or binary partition method~\cite{MR401486,MR1659922}, Avis and Fukuda's reverse search~\cite{MR1380066}, the proximity search method of Conte, Grossi, Marino, Uno and Versari~\cite{MR4502134}, the bubble language framework of Ruskey, Sawada, and Williams~\cite{MR2844089}, Williams' greedy algorithm~\cite{MR3126386}, and the permutation language framework of Hartung, Hoang, M\"utze and Williams~\cite{MR4391718}.

A particularly useful concept for developing efficient generation algorithms are Gray codes.
A \defi{combinatorial Gray code}~\cite{ruskey:2016,MR1491049,mutze:2023} for a class of objects is a listing of the objects such that every two consecutive objects in the list differ only by a local change.
Gray codes lend themselves very well to efficient generation, and often result in algorithms with small delay, sometimes even loopless algorithms.
For example, the spanning trees of a graph admit an edge exchange Gray code, i.e., they can be listed such that every two consecutive spanning trees differ in removing one edge from the current spanning tree and adding another edge from the graph to obtain the next spanning tree; see Figure~\ref{fig:span-gc}.
Furthermore, such a Gray code can be computed with amortized delay~$\cO(1)$~\cite{smith_1997} (see also \cite[Sec.~7.2.1.6]{MR3444818}).

\begin{figure}[h!]
\includegraphics[page=2]{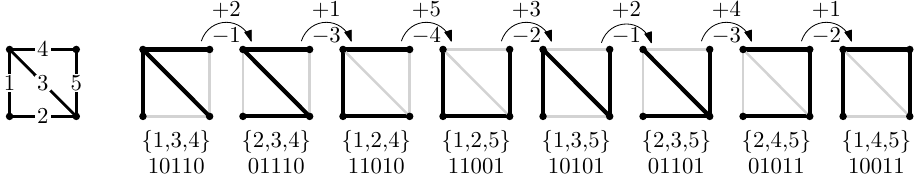}
\caption{Edge exchange Gray code for the spanning trees of the diamond graph on the top.
Below each spanning tree is the subset of tree edges and its indicator vector.}
\label{fig:span-gc}
\end{figure}

There is a trivial connection between combinatorial generation and optimization:
For a given class of objects~$X$, we can compute a minimizer of~$f(x)$ by exhaustively generating~$X$ and computing~$f(x)$ for each of the objects~$x$ in the list (for this the list need not be stored).
In terms of running time, this approach is prohibitively expensive in many applications, as the size of~$X$ is often exponential in some parameter.
For example, a minimum weight spanning tree in an $n$-vertex graph can be computed in time polynomial in~$n$, but the number of spanning trees is typically exponential in~$n$.

\begin{table}
\caption{Overview of results derived from our generation framework.
For each class of combinatorial objects we show the local change operation between consecutively generated objects, the corresponding optimization problem, and the resulting (worst-case) delay per generated object, highlighted if superior to previous results.
Earlier algorithms are listed in the penultimate column, and they are sometimes incomparable (worst-case delay vs.\ amortized delay, Gray code vs. non-Gray code).
None of the earlier algorithms computes a Hamilton path on the polytope, with the only exception being the results~\cite{smith_1997,MR4473269} on spanning trees and matroid bases. 
The last column gives pointers to later sections with detailed derivations.
The vector $c\in\mathbb{Z}^n$ is an arbitrary integer-valued cost vector with maximum absolute value~$|c|$.
}
\label{tab:appl}
\renewcommand{\arraystretch}{1.1}
\setlength\tabcolsep{2pt}
\tiny
\makebox[0cm]{ 
\begin{tabular}{|p{3mm}|L{25mm}|L{26mm}|L{45mm}|L{32mm}|L{42mm}|p{10mm}|}
\hhline{~------}
\multicolumn{1}{c|}{} & {\bf Objects} & {\bf Local change} & {\bf Optimization problem} & {\bf Delay} & {\bf Previous work} & {\bf Ref.} \\ \hhline{-------}
\multirow{4}{*}{\rotatebox{90}{\bf Other}}
 & vertices of a 0/1-polytope & edge move & linear programming with weights $\{-1,0,1\}$: $t_{\upright{LP}}$ & \cellcolor{red!25} $\cO(t_{\upright{LP}}\log n)$ & \cite{MR1659922}: $\cO(t_{\upright{LP}}\,n)$ delay & Cor.~\ref{cor:appl-poly} \\ \hhline{~~-----}
 & $c$-optimal vertices of a 0/1-polytope & edge move & linear programming: $t_{\upright{LP}}$ & \cellcolor{red!25} $\cO(t_{\upright{LP}}\poly(\log n))$ & \cite{MR1659922}: $\cO(t_{\upright{LP}}\,n)$ delay & Cor.~\ref{cor:appl-poly-c} \\ \hhline{~------}
 & feasible solutions to a knapsack problem & & knapsack with profits $\{-1,0,1\}$: $\cO(n)$ (sort weights in preprocessing) & $\cO(n \log n)$ & \cite{MR2900417}: $\cO(1)$ amortized delay Gray code & \\ \hhline{~------}
 & matroid bases & element exchange & minimum weight basis with weights $\{-2,-1,0,1,2\}$: $t_{\upright{LO}}$ & $\cO(t_{\upright{LO}}\log n)$ & \cite{MR1730351}: amortized delay \newline \cite{MR4473269}: delay Gray code & \\ \hhline{~~-----}
 & $c$-optimal matroid bases & same weight element exchange & minimum weight basis: $t_{\upright{LO}}$ & \cellcolor{red!25} $\cO(t_{\upright{LO}}\log n)$ & & \\ \hhline{~------}
 & matroid independent sets & add/remove/ exchange element & minimum weight basis with weights~$\{-2,-1\}$: $t_{\upright{LO}}$ & \cellcolor{red!25} $\cO(t_{\upright{LO}}\log n)$ & & \\ \hhline{~------}
 & matroid intersection maximum ind.\ sets & & minimum weight matroid intersection: $t_{\upright{LO}}$ & \cellcolor{red!25} $\cO(t_{\upright{LO}}\log n)$ & & \\ \hline
\multirow{4}{*}{\rotatebox{90}{\bf Graph}}
 & spanning trees & edge exchange & minimum spanning tree with weights $\{-2,-1,0,1,2\}$: $\cO(m)$ & $\cO(m\log n)$ & \cite{MR1146706}: $\cO(1)$ amortized delay \newline \cite{smith_1997}: $\cO(1)$ amort.\ delay Gray code \newline \cite{MR4473269}: $\cO(m \log n(\log \log n)^3)$ delay Gray code & Cor.~\ref{cor:appl-tree} \\ \hhline{~~-----}
 & $c$-optimal spanning trees & same weight edge exchange & minimum spanning tree: $\cO(m\alpha(m,n))$~\cite{MR1866456} & $\cO(m\alpha(m,n)\log n)$ & \cite{MR2754637}: $\cO(m\log n)$ amortized delay & Cor.~\ref{cor:appl-tree-c} \\ \hhline{~------}
 & forests & add/remove/ exchange edge & spanning forest computation \textrightarrow{} graph search: $\cO(m+n)$ & \cellcolor{red!25} $\cO((m+n) \log n)$ & & \\ \hhline{~------}
 & matchings & alternating $\leq\!3$-path exchange & maximum matching: $\cO(m\sqrt n)$~\cite{DBLP:conf/focs/MicaliV80} & \cellcolor{red!25} $\cO(m\sqrt n \log n)$ & & Cor.~\ref{cor:appl-match} \\ \hhline{~~-----}
 & maximum matchings & alternating path/cycle exchange & maximum weight matching with weights $\{m-1,m,m+1\}$: $\cO(m\sqrt n \log (n))$~\cite{MR3763658} & \cellcolor{red!25} $\cO(m\sqrt{n}(\log n)^2)$ & & Cor.~\ref{cor:appl-match-max} \\ \hhline{~~-----}
 & $c$-optimal matchings & \multirow{2}{20mm}{$c$-balanced alternating path/cycle exchange} & maximum weight matching: $\cO(m\sqrt n \log (n|c|))$~\cite{MR3763658} & \cellcolor{red!25} $\cO(m\sqrt{n}\log(n|c|)\log n)$ & & Cor.~\ref{cor:appl-match-c} \\ \hhline{~~~--~~}
 & & & maximum weight matching: $\cO(mn+n^2\log n)$~\cite{MR3744699} & \cellcolor{red!25} $\cO((mn+n^2\log n)\log n)$ & & \\ \hline
\multirow{8}{*}{\rotatebox{90}{\bf Bipartite graph}}
 & matchings & alternating $\leq\!3$-path exchange & maximum bipartite matching: $\cO(m\sqrt n)$~\cite{MR337699} & \cellcolor{red!25} $\cO(m\sqrt n \log n)$ & & \\ \hhline{~~-----}
 & maximum matchings & alternating path/cycle exchange & maximum weight bipartite matching with weights $\{m-1,m,m+1\}$: $\cO(m\sqrt{n}\log n)$~\cite{MR3205301} & $\cO(m\sqrt{n}(\log n)^2)$ & \cite{MR1651024}: $\cO(n)$ delay & \\ \hhline{~~-----}
 & $c$-optimal matchings & \multirow{2}{20mm}{$c$-balanced alternating path/cycle exchange} & maximum weight bipartite matching: $\cO(m\sqrt{n}\log(|c|))$~\cite{MR3205301} & $\cO(m\sqrt{n}\log(|c|)\log n)$ & \cite{MR1170948}: $\cO(mn)$ delay for $c$-optimal \emph{perfect} matchings & \\ \hhline{~~~--~~}
 & & & maximum weight bipartite matching: $\cO(mn+n^2\log \log n)$~\cite{MR2087939} & $\cO((mn\!+\!n^2\log \log n)\log n)$ &  & \\ \hhline{~------}
 & vertex covers & connected symmetric difference exchange & minimum bipartite vertex cover: $\cO(m\sqrt n)$~\cite{MR337699} & \cellcolor{red!25} $\cO(m\sqrt n \log n)$ & & \\ \hhline{~~-----}
 & minimum vertex covers & connected symmetric difference exchange & maximum flow with capacities $\{n-1,n,n+1\}$: $\cO(mn)$~\cite{MR3210838} & \cellcolor{red!25} $\cO(mn\log n)$ & & \\ \hhline{~~-----}
 & $c$-optimal vertex covers & connected symmetric difference exchange & maximum flow: $\cO(mn)$~\cite{MR3210838} & \cellcolor{red!25} $\cO(mn\log n)$ & & \\ \hline
 \multirow{4}{*}{\rotatebox{90}{\bf Poset}}
 & antichains & connected symmetric difference exchange & maximum antichain \textrightarrow{} maximum bipartite matching: $\cO(n^{2.5}/\sqrt{\smash[b]{\log n}})$ \cite{MR1095712} (cf.~\cite{MR2079151}) & \cellcolor{red!25} $\cO(n^{2.5}\sqrt{\smash[b]{\log n}})$ & & \\ \hhline{~~-----}
 & maximum antichains & connected symmetric difference exchange & MWCCG with weights $\{n-1,n,n+1\}$: $\cO(n^4)$~\cite{MR1755057} & \cellcolor{red!25} $\cO(n^4 \log n)$ & & \\ \hhline{~~-----}
 & $c$-optimal antichains & connected symmetric difference exchange & maximum weight clique in cocomparability graph (MWCCG): $\cO(n^4)$~\cite{MR1755057} & \cellcolor{red!25} $\cO(n^4 \log n)$ & & \\ \hhline{~------}
 & ideals & connected symmetric difference exchange & MWCCG with weights $\{-1,1\}$: $\cO(n^4)$~\cite{MR1755057} & $\cO(n^4 \log n)$ & [Squire~1995] (see~\cite{ruskey_2003}): $\cO(\log n)$ amortized delay \newline \cite{MR1267190,MR1828422}: $\cO(n)$ amortized delay Gray code & \\ \hhline{~~-----}
 & $c$-optimal ideals & connected symmetric difference exchange & MWCCG: $\cO(n^4)$~\cite{MR1755057} & \cellcolor{red!25} $\cO(n^4 \log n)$ & & \\ \hline
\end{tabular}
}
\end{table}

\torsten{Should try to retrieve the original Squire paper. Asked him about it.}
\torsten{We should add max flows as done by Bussieck/Luebbecke and look at the improvement.}
\torsten{What about 0/1-polytopes associated with acyclic orientations?}
\torsten{Cut polytopes: \url{https://conservancy.umn.edu/bitstream/handle/11299/158678/1/Ganguly.pdf}.}
\torsten{Knapsack polytope has been heavily studied.}
\torsten{Shortest paths/Dijkstra stuff}
\torsten{We should start thinking about maximal substructures of graphs}
\torsten{What we call hereditary some people call `monotone': \cite{MR2503771}}

\subsection{Our contribution}

In this work, we establish a nontrivial connection between combinatorial generation and optimization that goes in the opposite direction.
Specifically, we show that if the optimization problem~$\min_{x\in X}f(x)$ can be solved efficiently, then this directly yields an efficient generation algorithm for~$X$.
More precisely, the delay for the resulting generation algorithm is larger than the running time of an arbitrary optimization algorithm only by a logarithmic factor.
The optimization algorithm is used as a black box inside the generation algorithm, and in this way we harness the powerful machinery of combinatorial optimization for the purpose of combinatorial generation.
Furthermore, the generated listings of objects correspond to a Hamilton path on the skeleton of a 0/1-polytope that is associated naturally with the combinatorial objects, i.e., we obtain a Gray code listing.
Additional notable features of our generation algorithm are: the algorithm is conceptually simple and operates greedily, it uses only elementary data structures, it is easy to implement, and it contains several tunable parameters that can be exploited for different applications.
We thus provide a versatile algorithmic framework to systematically solve the combinatorial generation problem for a large variety of different classes of combinatorial objects, which comes with strong guarantees for the delay and the closeness between consecutively generated objects.
Table~\ref{tab:appl} summarizes the most important concrete results obtained from our framework, and those will be discussed in more detail in Section~\ref{sec:table}.
Informally speaking, we add a new powerful `hammer' to the toolbox of combinatorial generation research, which is forged by combining algorithmic, combinatorial and polyhedral ideas.

In this paper, we focus on presenting the main ingredients of our framework that connects combinatorial optimization to combinatorial generation, and we illustrate the method by deriving several new generation algorithms for a number of combinatorial objects based on graphs, matroids, posets and polytopes.
This paper will be followed by others in which this new paradigm is exploited further in various directions (implementations, computational studies, improved results for special cases, etc.).

\subsection{Encoding of objects by bitstrings}

To run our generation algorithm, we rely on a unified method to encode the various different classes~$X$ of combinatorial objects.
For this we use a set of bitstrings~$X\seq\{0,1\}^n$ of length~$n$, which lends itself particularly well for computer implementation.
An important notion in this context is an indicator vector.
Given a subset $S\seq [n]:=\{1,\ldots,n\}$, the \defi{indicator vector} $\indicator_S \in \{0,1\}^n$ is defined as
\[(\indicator_S)_i := \begin{cases}
 1 & \text{if } i \in S, \\
 0 & \text{if } i\notin S.
\end{cases}\]
For example, the set $S=\{1,4,5\}\seq [6]$ has the indicator vector $\indicator_S=(1,0,0,1,1,0)$.
The set of bitstrings~$X\seq\{0,1\}^n$ that is used to encode the combinatorial objects of interest is simply the set of all corresponding indicator vectors.
Specifically, combinatorial objects based on graphs are encoded by considering the corresponding subsets of vertices or edges, and by their indicator vectors.
For example, indicator vectors for the spanning trees or matchings of a graph~$H$ have length~$m$, where $m$ is the number of edges of~$H$; see Figure~\ref{fig:span-gc}.
Similarly, indicator vectors for ideals or antichains of a poset~$P$ have length~$n$, where $n$ is the number of elements of~$P$.
Note that this encoding is based on a particular ordering of the ground set.
In other words, changing the ordering of the ground set corresponds to permuting the entries of all bitstrings in~$X$.

\begin{figure}[b!]
\centering
\begin{tabular}{ccc}
\includegraphics[height=4cm,page=1]{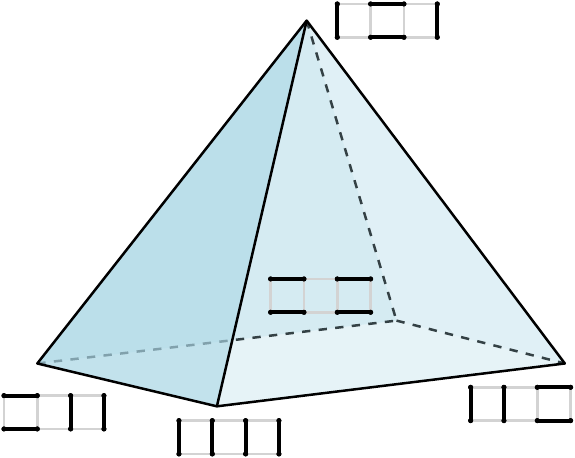} &
\includegraphics[height=4cm,page=2]{poly} &
\includegraphics[height=4cm,page=3]{poly} \\
(a) & (b) & (c)
\end{tabular}
\caption{Three combinatorial 0/1-polytopes (after appropriate projections into 3-dimensional space and under combinatorial equivalence):
(a) perfect matching polytope of the $2\times 4$ grid graph;
(b) vertex cover polytope of the triangle graph;
(c) base polytope of the uniform matroid of 2-element subsets of the ground set~$\{1,2,3,4\}$.
}
\label{fig:poly}
\end{figure}

\begin{table}[b!]
\caption{Examples of 0/1-polytopes that encode local change operations on combinatorial objects through their skeleton.}
\label{tab:01poly}
\renewcommand{\arraystretch}{1.1}
\setlength\tabcolsep{3pt}
\makebox[0cm]{ 
\begin{tabular}{|L{20mm}|L{35mm}|L{45mm}|L{35mm}|p{14mm}|}
\hline
{\bf Parameter} & {\bf Polytope} & {\bf Vertex set~$X$} & {\bf Edges (=flips)} & {\bf Ref.} \\ \hline
integer $n$ & $n$-dimensional hypercube & $\{0,1\}^n$, i.e., bitstrings of length~$n$ & flip a single bit & \\ \hline
integer $n$ & Birkhoff polytope & $n\times n$ permutation matrices & multiplication with a cycle & \cite{MR1311028} \\ \hline
integers $n,k$ & uniform matroid base polytope & $(n,k)$-combinations & transpositions & \cite{MR510371} \\ \hline
connected graph~$H$ & spanning tree polytope & indicator vectors of spanning trees of~$H$ & edge exchange & \cite{MR510371} \\ \hline
graph~$H$ & matching polytope & indicator vectors of matchings in~$H$ & alternating path/cycle exchange & \cite{MR371732} \\ \hline
graph~$H$ & perfect matching polytope & indicator vectors of perfect matchings in~$H$ & alternating cycle exchange & \cite{MR371732} \\ \hline
graph~$H$ & stable set polytope & indicator vectors of stable (=independent) sets in~$H$ & connected symmetric difference exchange & \cite{MR371732} \\ \hline
graph~$H$ & vertex cover polytope & indicator vectors of vertex covers in~$H$ & connected symmetric difference exchange & \cite{MR510371} \\ \hline
poset~$P$ & chain polytope & indicator vectors of antichains in~$P$ & connected symmetric difference exchange & \cite{MR824105,math7050381} \\ \hline
poset $P$ & order polytope & indicator vectors of ideals in~$P$ & connected symmetric difference exchange & \cite{MR824105,math7050381} \\ \hline
\end{tabular}
}
\end{table}

\torsten{The Birkhoff polytope is the perfect matching polytope of~$K_{n,n}$.}

\subsection{Combinatorial 0/1-polytopes}

A crucial feature of this encoding via bitstrings is that it allows to equip the combinatorial objects with a natural \defi{polytope} structure.
This idea is the cornerstone of polyhedral combinatorics and has been exploited extensively in combinatorial optimization.
Specifically, for a given set of bitstrings~$X\seq \{0,1\}^n$, we consider the convex hull
\[\conv(X):= \Big\{\sum\nolimits_{x\in X} \nu_x x \;\Big\vert\; \nu_x \in [0,1] \text{ for all } x\in X \text{ and } \sum\nolimits_{x\in X} \nu_x=1 \Big\}, \]
which is a 0/1-polytope with vertex set~$X$.
More generally, the combinatorial structure of~$X$ is encoded in the face structure of~$\conv(X)$.
In particular, the edges of the polytope~$\conv(X)$ correspond to local changes between the objects from~$X$ that are the endpoints of this edge.
For example, if $X$ is the set of indicator vectors of spanning trees of a graph~$H$, then the edges of~$\conv(X)$ connect pairs of spanning trees that differ in an edge exchange.
Or, if $X$ is the set of indicator vectors of matchings of~$H$, then the edges of~$\conv(X)$ connect pairs of matchings that differ in an alternating path or cycle.
Some examples of such combinatorial 0/1-polytopes are visualized in Figure~\ref{fig:poly}, and more are listed in Table~\ref{tab:01poly}.
Note that the dimension~$n$ quickly exceeds~3 even for moderately large examples, so the figure shows combinatorially equivalent projections into 3-dimensional space (where the coordinates are not 0/1 anymore).

We are mostly interested in the \defi{skeleton} of the polytope under consideration, i.e., the graph defined by its vertices and edges; see Figure~\ref{fig:span-skel}.
In this context it makes sense to refer to bitstrings from~$X$ as vertices.
As mentioned before, the edges of the skeleton capture local change operations between the combinatorial objects that are represented by the vertices.
Such graphs are sometimes called \defi{flip graphs} in the literature, where \defi{flip} is a general term for a local change operation (flip graphs can be defined without reference to polytopes).

\begin{figure}[b]
\centering
\makebox[0cm]{ 
\includegraphics[page=4]{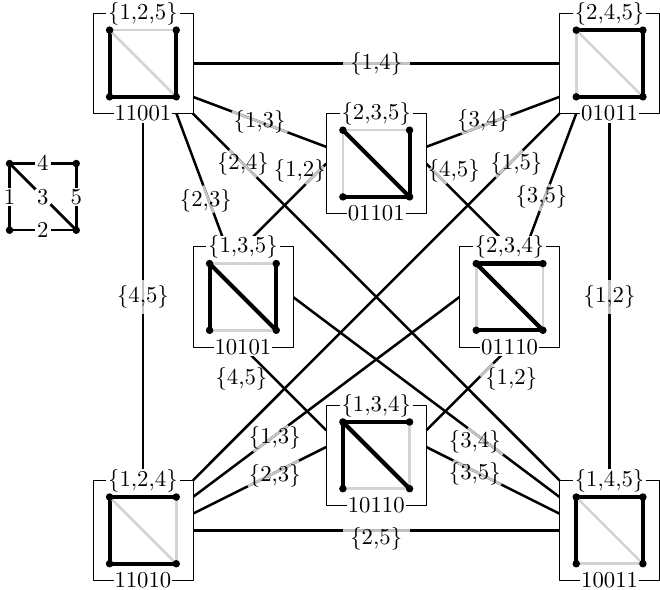}
}
\caption{Skeleton of the spanning tree polytope of the diamond graph.
The edges of the skeleton are labeled by edge exchanges.}
\label{fig:span-skel}
\end{figure}

The most notable feature of the generation algorithm presented in this paper is that the listings of objects~$X$ generated by the algorithm correspond to a \defi{Hamilton path} on the skeleton of~$\conv(X)$, i.e., a path that visits every vertex of the polytope exactly once.
As a result, every two consecutive objects in the listing differ by a local change (=flip), i.e., we obtain a Gray code.

\subsection{The basic algorithm}
\label{sec:basic-algo}

To describe our algorithm, we define for two distinct bitstrings $x,y\in X\seq\{0,1\}^n$ the quantity $\lambda(x,y):=\max\{i\in[n]\mid x_i\neq y_i\}$.
In words, $\lambda(x,y)$ is the largest index in which~$x$ and~$y$ differ.
Equivalently, the longest suffix in which~$x$ and~$y$ agree has length~$n-\lambda(x,y)$.
For example, we have $\lambda(1010110,0111110)=4$ and $\lambda(01100,11100)=1$.
For two bitstrings~$x,y\in\{0,1\}^n$ we define the \defi{Hamming distance} of~$x$ and~$y$ as $d(x,y):=|\{i\in[n]\mid x_i\neq y_i\}|$.
In words, this is the number of positions in which~$x$ and~$y$ differ.
For example, we have $d(1010110,0111110)=3$ and $d(01100,11100)=1$.

With these definitions at hand, we are in position to describe our basic generation algorithm, shown in Algorithm~P as pseudocode.
Algorithm~P takes as input a set of bitstrings~$X\seq\{0,1\}^n$ (possibly given implicitly; recall Table~\ref{tab:01poly}) and it computes a Hamilton path on the skeleton of the 0/1-polytope~$\conv(X)$, starting at an initial vertex~$\tx$ that is also provided as input.
The current vertex~$x$ is visited in step~P2, and subsequently the next vertex to be visited is computed in steps~P3 and~P4.
Specifically, in step~P3 we compute the length~$\beta$ of the shortest prefix change between~$x$ and an arbitrary unvisited vertex~$y\in X-x$, where $X-x:=X\setminus\{x\}$.
In step~P4 we consider vertices~$y$ that differ from~$x$ in a prefix of length~$\beta$, i.e., $\lambda(x,y)=\beta$, and among those we select the ones with minimum Hamming distance from~$x$ into the set~$N$.
We will show that all vertices in~$N$ are actually unvisited.
In step~P5, one of the vertices~$y\in N$ is chosen as the next vertex to be visited by the algorithm.
If the set~$N$ contains more than one element, then we have freedom to pick an arbitrary vertex~$y\in N$.
In concrete applications, one will usually equip Algorithm~P with a \defi{tiebreaking rule}, which for a given set~$N\seq X$ and the current state of the algorithm selects an element from~$N$.

\begin{algo}{Algorithm~P}{Traversal of 0/1-polytope by shortest prefix changes}
For a set $X\seq\{0,1\}^n$, this algorithm greedily computes a Hamilton path on the skeleton of the 0/1-polytope~$\conv(X)$, starting from an initial vertex~$\tx$.
\begin{enumerate}[label={\bfseries P\arabic*.}, leftmargin=8mm, noitemsep, topsep=3pt plus 3pt]
\item{} [Initialize] Set $x \gets \tx$.
\item{} [Visit] Visit~$x$.
\item{} [Shortest prefix change] 
Terminate if all vertices of~$X$ have been visited.
Otherwise compute the length~$\beta$ of the shortest prefix change between~$x$ and an arbitrary unvisited vertex~$y\in X-x$, i.e., $\beta\gets \min_{y\in X-x\,\wedge\,y \text{ unvisited}} \lambda(x,y)$.
\item{} [Closest vertices] Compute the set~$N$ of vertices~$y$ with $\lambda(x,y)=\beta$ of minimum Hamming distance from~$x$, i.e., $N\gets \argmin_{y\in X-x\,\wedge\,\lambda(x,y)=\beta}d(x,y)$.
\item{} [Tiebreaker+update~$x$] Pick an arbitrary vertex $y\in N$, set $x\gets y$ and goto~P2.
\end{enumerate}
\end{algo}

\begin{figure}
\centering
\makebox[0cm]{ 
\includegraphics[page=5]{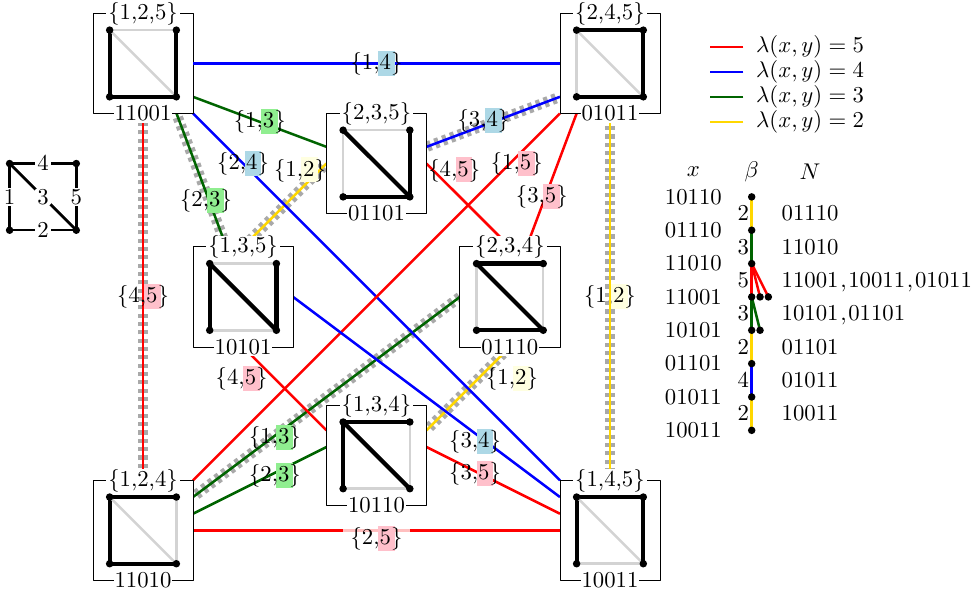}
}
\caption{Run of Algorithm~P on the skeleton from Figure~\ref{fig:span-skel}.
Edges are colored according to $\lambda$-values.
The computed Hamilton path is drawn dashed, and it corresponds to the edge exchange Gray code from Figure~\ref{fig:span-gc}.
}
\label{fig:span-algo}
\end{figure}

The listing of spanning trees shown in Figure~\ref{fig:span-gc} is a possible output of Algorithm~P.
The corresponding Hamilton path on the skeleton of the spanning tree polytope is highlighted by a dashed line in Figure~\ref{fig:span-algo}.
This figure colors the edges~$(x,y)$ of the skeleton according to their~$\lambda$-value $\lambda(x,y)$, and it shows the state of the variables~$x$, $\beta$ and~$N$ through each iteration of the algorithm.
Note that in the third and fourth iteration the set $N$ has more than one element, so ties have to be broken.
Observe how the algorithm greedily gives priority to traversing edges with small $\lambda$-value (a short prefix change) compared to large $\lambda$-value (a long prefix change).

The following fundamental theorem is the basis of our approach. 

\begin{theorem}
\label{thm:algoP-poly}
For every set $X\seq\{0,1\}^n$, every tiebreaking rule and every initial vertex~$\tx$, Algorithm~\upright{P} computes a Hamilton path on the skeleton of~$\conv(X)$ starting at~$\tx$.
\end{theorem}

Algorithm~P has a number of striking features:
\begin{itemize}[itemsep=0ex,parsep=0ex,leftmargin=0ex,itemindent=0ex,labelsep=1ex,labelwidth=-2ex]
\item
It works for \emph{every} set of bitstrings~$X\seq\{0,1\}^n$, i.e., for every 0/1-polytope.
\item
It works for \emph{every} tiebreaking rule used in step~P5.
In an actual implementation, we would directly compute one particular vertex from~$N$, instead of computing all of them and then selecting one.
The reason why the set~$N$ appears in the pseudocode is to emphasize the freedom we have in choosing a tiebreaking rule according to the needs of the application.
\item
It works for \emph{every} initial vertex~$\tx$, which creates room for exploitation in different applications.
\item
It works for \emph{every} ordering of the ground set~$[n]$.
In fact, we could state the algorithm in an even more general form and base the computation of~$\lambda$ on an arbitrary total order of~$[n]$.
This freedom can be very helpful in various applications.
For example, when computing the spanning trees of a graph~$H$, different orderings of the edges of~$H$ will result in different listings of the spanning trees (even for the same initial spanning tree and the same tiebreaking rule).
\item
The ordering of bitstrings from~$X$ produced by the algorithm has the so-called \defi{genlex} property, i.e., bitstrings with the same suffix appear consecutively.
In other words, the algorithms visits vertices ending with~0 before all vertices ending with~1, or vice versa, and this property is true recursively within each block of vertices with the same suffix removed.
\item
The algorithm performs no polyhedral computations at all.
In particular, there are no references to the edges or higher-dimensional faces of the polytope~$\conv(X)$ in the algorithm.
Instead, it relies on combinatorial properties of the bitstrings in~$X$, namely $\lambda$-values and Hamming distances.
The fact that the resulting listing is a Hamilton path on the skeleton of the polytope~$\conv(X)$ is a consequence of these combinatorial properties, but not built into the  algorithm.
\item
The algorithm computes a Hamilton path on the skeleton of the 0/1-polytope~$\conv(X)$, i.e., we obtain a Gray code listing of the objects encoded by~$X$ with closeness guarantees between consecutive objects; recall Table~\ref{tab:01poly}.
Observe that the algorithm does not traverse \emph{arbitrary} edges of the 0/1-polytope, but only edges with minimum Hamming distance (as mentioned before, the algorithm does not even `know' that these are polytope edges).
As a consequence, these closeness guarantees can be strengthened considerably in many cases.
For example, in the matching polytope any two matchings that differ in an alternating path or cycle are adjacent, but Algorithm~P will only traverse edges corresponding to alternating paths of length~$\leq 3$.
Similarly, in the Birkhoff polytope any two permutations that differ in a single cycle are adjacent, but Algorithm~P will only traverse edges corresponding to cycles of length~2 (i.e., transpositions).
\item
The computation of~$\beta$ in step~P3 requires the qualification `$y$ unvisited', which seems to indicate that we need to store all previously visited vertices, which would require exponential space and would therefore be very detrimental.
However, we show that with some simple additional data structures, we can make the algorithm \defi{history-free}, so no previously visited vertices need to be stored at all, but instead the algorithm uses only~$\cO(n)$ extra space (in addition to the input).
This history-free implementation of Algorithm~P is described as Algorithm~P\sss{} in Section~\ref{sec:algoP}.
\item
The algorithm is straightforward to implement with few lines of code.
\item
As we shall explain in the next section, the computations in steps~P3 and~P4 can be done by solving one or more instances of a linear optimization problem over~$X$.
Consequently, we can use \emph{every} optimization algorithm as a black box inside Algorithm~P (or its history-free variant).
The running time of the generation algorithm then depends on the time needed to solve the optimization problem, and the resulting delay is larger than the time needed to solve the optimization problem only by a $\log n$ factor.
Specificially, the $\log n$ factor comes from solving several instances of the optimization problem, and doing binary search.
\item
Algorithm~P thus provides a \emph{general and automatic} way to exploit an optimization algorithm for the purpose of generation.
With each new optimization algorithm we automatically obtain a baseline against which to compare any new generation algorithm for that particular class of objects.
Also, the delays for concrete problems listed in Table~\ref{tab:appl} obtained from our work may improve as soon as someone finds an improved optimization algorithm.
\end{itemize}

We also mention the following two shortcomings of Algorithm~P:
\begin{itemize}[itemsep=0ex,parsep=0ex,leftmargin=0ex,itemindent=0ex,labelsep=1ex,labelwidth=-2ex]
\item
For unstructured classes of objects, for example the set of all bitstrings of length~$n$, or the set of all $n\times n$ permutation matrices, there are sometimes faster (often loopless) and more direct generation methods available than what we obtain via optimization.
While Algorithm~P still works in those cases, its performance shines primarily for classes of objects that arise from certain constraints, for example from a graph, a matroid, a poset, or from a cost function.

\item
In general, Algorithm~P does not compute a Hamilton \emph{cycle}, but only a Hamilton path on the 0/1-polytope (cf.\ Section~\ref{sec:related} below).
There are classes of objects and choices of the tiebreaking rule so that the last vertex of the computed path is adjacent to the first one, but these situations are out of our control in general.
This is because the algorithm operates entirely locally on the skeleton and `forgets' the location of the initial vertex~$\tx$.
\end{itemize}

\subsection{Reduction to classical linear optimization}
\label{sec:reduction}

A classical problem in combinatorial optimization, which includes for example minimum spanning tree, maximum weight (perfect) matching, maximum stable set, minimum vertex cover and many others, is the \defi{linear optimization} problem, defined as follows:
\begin{enumerate}[label=\mybox{LO},leftmargin=12mm, noitemsep, topsep=1pt plus 1pt]
\item Given a set~$X\seq\{0,1\}^n$ and a weight vector~$w\in \mathbb{R}^n$, compute an element in
\[N:=\argmin_{y\in X} w\cdot y,\]
or decide that this problem is infeasible, i.e., $N=\emptyset$.
\end{enumerate}

The key insight of our paper is that the computation of~$\beta$ in step~P3 and of the set~$N$ in step~P4 of Algorithm~P can be achieved by solving one or more instances of the following variant of the problem~LO, which we refer to as \defi{linear optimization with prescription}.
\begin{enumerate}[label=\mybox{LOP},leftmargin=12mm, noitemsep, topsep=1pt plus 1pt]
\item Given a set~$X\seq\{0,1\}^n$, a weight vector~$w\in \mathbb{R}^n$, and disjoint sets $P_0,P_1\seq [n]$, compute an element in
\[N:=\argmin_{y\in X\,\wedge\, y_{P_0}=0 \,\wedge\, y_{P_1}=1} w\cdot y,\]
or decide that this problem is infeasible, i.e., $N=\emptyset$.
\end{enumerate}
The notation $y_{P_b}=b$ for $b\in\{0,1\}$ is a shorthand for $y_i=b$ for all $i\in P_b$.
In words, the value of~$y$ is prescribed to be~0 or~1 at the coordinates in~$P_0$ and~$P_1$, respectively.

We explain the reduction to the problem~LOP for the computation of the set
\[N:=\argmin_{y\in X-x\,\wedge\,\lambda(x,y)=\beta}d(x,y)\]
in step~P4 of Algorithm~P.
Given $x$ and~$\beta\in[n]$, we define the weight vector
\begin{equation}
\label{eq:wi}
w_i:=\begin{cases}
+1 & \text{if } x_i=0, \\
-1 & \text{if } x_i=1, \\
\end{cases}
\end{equation}
and the sets
\[
P_b:=\{\beta\mid x_\beta=1-b\}\cup \{i>\beta\mid x_i=b\} \text{ for } b\in\{0,1\},
\]
and we use $w$ and $P_0,P_1$ as input for the optimization problem~LOP.
The definition of $P_0,P_1$ ensures that all feasible $y\in X$ satisfy $y_i=x_i$ for $i>\beta$ and $y_\beta=1-x_\beta$, which ensures that they all satisfy $\lambda(x,y)=\beta$, in particular $y\neq x$, i.e., $y\in X-x$.
Furthermore, the definition of $w$ implies that $d(x,y)=w\cdot (y-x)$, and since $x$ is fixed, minimization of~$d(x,y)$ is the same as minimization of $w\cdot y$.
We consequently obtain that
\[N=\argmin_{y\in X-x\,\wedge\,\lambda(x,y)=\beta}d(x,y)=\argmin_{y\in X\,\wedge\,y_{P_0}=0\,\wedge\,y_{P_1}=1}w\cdot y.\]

The computation of~$\beta$ in step~P3 of our algorithm can be done similarly, however with an extra ingredient, namely binary search.
The binary search causes the $\log n$ factor in the delay.
Specifically, we obtain that if problem~LOP over~$X$ can be solved in time~$t_{\upright{LOP}}$, then Algorithm~P runs with delay $\cO(t_{\upright{LOP}}\log n)$ (Theorem~\ref{thm:algoPLOP-time}).
We also provide a variant of this reduction for generating only the elements in~$X$ that are optimal w.r.t.\ some cost vector~$c \in \mathbb{Z}^n$ (Theorem~\ref{thm:algoPLOP-time-c}).

In many cases (e.g., spanning trees, matchings, perfect matchings, etc.) the problem~LOP with weights~$w$ as in~\eqref{eq:wi} reduces to the problem~LO directly (by removing or contracting prescribed edges).
In those cases we obtain the same results as before, but with $t_{\upright{LO}}$ instead of $t_{\upright{LOP}}$.
In fact, by weight amplification the problem~LOP can always be reduced to the problem~LO (Lemma~\ref{lem:LOP-LO}), albeit at the expense of increasing the weights, which may worsen the running time of the optimization algorithm.

\subsection{Applications}
\label{sec:table}

It turns out that Algorithm~P is very versatile, and allows efficient generation of a large variety of combinatorial objects based on graphs, matroids, posets and polytopes in Gray code order.
Table~\ref{tab:appl} lists these results in condensed form, and in the following we comment on the most important entries in the table.

\subsubsection{Vertex enumeration of 0/1-polytopes}

The \defi{vertex enumeration problem} for polytopes is fundamental in discrete and computational geometry; see the surveys~\cite{MR571811} and~\cite{MR716120}, and the more recent work~\cite{MR1134355,MR2503771,MR2823189,MR4078802}.\footnote{Throughout this paper we use the term `generation' instead of `enumeration', but since `vertex enumeration' is a standard term in the polytope community, we stick to it here.} 
Given a linear system $Ax\leq b$ of inequalities, where $A\in \mathbb{R}^{m\times n}$ and $b\in\mathbb{R}^m$, the problem is to generate all vertices of the polytope $P:=\{x\in\mathbb{R}^n\mid Ax\leq b\}$.
For general polytopes, this problem can be solved by the \emph{double description method}~\cite{MR0060202}, or by the \emph{reverse search method} due Avis and Fukuda~\cite{MR1174359}, with its subsequent improvement~\cite{MR1785299}.
For the special case where $P$ is a 0/1-polytope, i.e., all vertex coordinates are from~$\{0,1\}$, Bussieck and L\"ubbecke~\cite{MR1659922} described an algorithm for generating the vertices of~$P$ with delay~$\cO(t_{\upright{LP}}\,n)$, where $t_{\upright{LP}}$ is the time needed to solve the linear program (LP) $\min\{w\cdot x\mid Ax\leq b\}$.
Furthermore, the space required by their algorithm is the space needed to solve the~LP.
Their approach is an instance of the flashlight search or binary partition method~\cite{MR401486}.
Behle and Eisenbrand~\cite{DBLP:conf/alenex/BehleE07} described an algorithm for vertex enumeration of 0/1-polytopes that uses binary decision diagrams, which performs well in practice, but requires exponential space in the worst case.

Algorithm~P improves upon Bussieck and L\"ubbecke's algorithm in that the delay is reduced from~$\cO(t_{\upright{LP}}\,n)$ to~$\cO(t_{\upright{LP}}\log n)$ per generated vertex (Corollary~\ref{cor:appl-poly}).
As an additional feature, the vertices are visited in the order of a Hamilton path on the polytope's skeleton, whereas the earlier algorithm does not have this property.
The space required by our algorithm is only the space needed to solve the~LP.
We can also generate all vertices that are $c$-optimal w.r.t.\ some arbitrary integer-valued cost vector~$c\in\mathbb{Z}^n$, with delay~$\cO(t_{\upright{LP}}\poly(\log n))$ per vertex (Corollary~\ref{cor:appl-poly-c}).

\subsubsection{Bases and independent sets in matroids}

Algorithm~P allows generating the bases of a matroid, by computing a Hamilton path on the base polytope.
The delay is~$\cO(t_{\upright{LO}}\log n)$, where $n$ is the number of elements in the matroid, and $t_{\upright{LO}}$ is the time to solve the linear optimization problem (problem~LO defined in Section~\ref{sec:reduction}) for the bases of the matroid.
It is the first polynomial delay Gray code known for the weighted variant of the problem with cost vector~$c$.
Our algorithm can be specialized to generate bases of a graphic matroid, i.e., spanning trees of a graph (Corollaries~\ref{cor:appl-tree} and~\ref{cor:appl-tree-c}), or bases of the uniform matroid, i.e., fixed-size subsets of an $n$-element ground set, also known as combinations.
Analogous statements hold for independent sets of a matroid, and the corresponding specialization to the graphic case, namely forests of a graph.
The algorithm also applies to computing maximum independent sets in the intersection of two matroids.

\subsubsection{More graph objects}

We provide the first polynomial delay Gray codes for generating matchings, maximum matchings, and $c$-optimal matchings in graphs (Corollaries~\ref{cor:appl-match}, \ref{cor:appl-match-c}, and~\ref{cor:appl-match-max}, respectively).
The obtained listings correspond to a Hamilton path on the matching polytope.
We also provide the first polynomial delay Gray codes for generating vertex covers, minimum vertex covers, and $c$-optimal vertex covers in bipartite graphs.
The generated listings correspond to a Hamilton path on the vertex cover polytope.
The space required by our algorithms is only the space required for solving the corresponding optimization problems, i.e., polynomial in the size of the graph.
As the complement of a vertex cover is a stable (=independent) set, we also obtain efficient algorithms for generating stable sets, maximum stable sets and $c$-optimal stable sets in bipartite graphs, by traversing the stable set polytope.

\subsubsection{Poset objects}

We provide the first polynomial delay Gray codes for generating antichains, maximum antichains, $c$-optimal antichains, and $c$-optimal ideals of a poset.
The listings of objects correspond to Hamilton paths on the chain polytope and order polytope, respectively.
The space required by our algorithms is the space for solving the corresponding optimization problems, i.e., polynomial in the size of the poset.

\subsection{Related work}
\label{sec:related}

Theorem~\ref{thm:algoP-poly} implies that for every 0/1-polytope~$P$ and each of its vertices, there is a Hamilton path on the skeleton of~$P$ starting at that vertex.
Naddef and Pulleyblank~\cite{MR762893} proved a considerable strengthening of this result.
Specifically, they showed that the skeleton of every 0/1-polytope is either a hypercube and therefore Hamilton-laceable, or it is Hamilton-connected.
A graph is \defi{Hamilton-connected} if it admits a Hamilton path between any two distinct vertices, and a bipartite graph is \defi{Hamilton-laceable} if it admits a Hamilton path between any two vertices from distinct partition classes.
The Naddef-Pulleyblank construction is inductive, but it does not translate into an efficient algorithm, as the size of the skeleton is typically exponential; recall Table~\ref{tab:01poly}.

Conte, Grossi, Marino, Uno and Versari~\cite{MR4502134} recently presented a modification of reverse search called \emph{proximity search}, which yields polynomial delay algorithms for generating several classes of objects defined by inclusion-maximal subgraphs of a given graph, specifically maximal bipartite subgraphs, maximal degenerate subgraphs, maximal induced chordal subgraphs etc.
Their approach is based on traversing a suitable defined low-degree flip graph on top of the graph objects that is then traversed by backtracking.
A disadvantage of their approach is that it requires exponential space to store previously visited objects along the flip graph, unlike Algorithm~P which can be implemented with linear space.
Also, there is no polyhedral interpretation of their method, in particular the generated listings of objects are not Gray codes.

\subsection{Outline of this paper}

In Section~\ref{sec:prelim} we collect some notations that will be used throughout this paper.
Sections~\ref{sec:binary}--\ref{sec:opt} develop our theory that will yield Theorem~\ref{thm:algoP-poly} as a special case, and they establish a connection between the optimization problem and the generation problem.
Specifically, in Section~\ref{sec:binary} we present a simple greedy algorithm for computing a Hamilton path in a graph whose vertex set is a set of bitstrings.
In Section~\ref{sec:genlex} we consider Hamilton paths with special properties that this algorithm can compute successfully.
In Section~\ref{sec:prefix} we introduce a class of graphs called \emph{prefix graphs} which admit such Hamilton paths, and we provide a history-free implementation of the basic greedy algorithm for computing a Hamilton path in those graphs.
In Section~\ref{sec:opt} we show that skeleta of 0/1-polytopes are prefix graphs, and we specialize the earlier results and algorithms to the case of 0/1-polytopes.
We also reduce the computational problems for our history-free algorithm to instances of linear optimization with or without prescription.
The results listed in the last column of Table~\ref{tab:appl} are derived from our general theorems in Section~\ref{sec:appl}.
In Section~\ref{sec:duality}, we discuss the relation between our generation framework and the one presented in~\cite{MR4391718}, and we conclude in Section~\ref{sec:open} with some open questions.

\section{Preliminaries}
\label{sec:prelim}

For a graph~$G=(X,E)$ and a vertex~$x\in X$ we write $E(x)$ for the set of neighbors of~$x$ in~$G$.

We write $\varepsilon$ for the empty bitstring.
Also, for a bitstring~$x$ we write $\ol{x}$ for the complemented string.
For two bitstrings~$x$ and~$y$ we write $xy$ for the concatenation of~$x$ and~$y$.
Given a sequence~$L=x_1,\ldots,x_\ell$ of bitstrings and a bitstring~$y$, we also define $Ly:=x_1y,\ldots,x_\ell y$.
For a nonempty bitstring~$x$ we write $x^-$ for the string obtained from~$x$ by deleting the last bit.
Furthermore, for a set~$X\seq\{0,1\}^n$ we define~$X^-:=\{x^-\mid x\in X\}\seq\{0,1\}^{n-1}$ and for a sequence $L=x_1,\ldots,x_\ell$ of bitstrings we define $L^-:=x_1^-,\ldots,x_\ell^-$.
For a set $X\seq\{0,1\}^n$ and $b\in\{0,1\}$ we define $X^b:=\{x\in\{0,1\}\mid x_n=b\}$.
Similarly, for a sequence~$L$ of bitstrings of length~$n$ and $b\in\{0,1\}$ we write $L^b$ for the subsequence of strings in~$L$ that end with~$b$.

For a predicate $P$ and a real-valued function~$f$ defined on the set of objects~$X$ satisfying the predicate~$P$, we write $\argmin[P(x)\mid f(x)]:=\argmin_{P(x)}f(x)$ for the minimizers of~$f$ on~$X$.
This notation without subscripts is convenient for us, as the predicates we need are complicated, so printing these long expressions as subscripts would make the formulas too unwieldy.
Our new notation is typographically `dual' to the minimum values $\min \{f(x)\mid P(x)\}$.

\section{Binary graphs and a simple greedy algorithm}
\label{sec:binary}

A \defi{binary graph} is a graph $G=(X,E)$ such that $X\seq \{0,1\}^n$ for some integer~$n$.
In our applications, the graph~$G$ is typically given implicitly, using a description of size polynomial in~$n$, whereas the size of~$G$ will be exponential in~$n$.
Examples for such graphs are the skeletons of the polytopes listed in Table~\ref{tab:01poly}, which are described by the parameter listed in the first column.
In the setting of Gray codes, one routinely obtains binary graphs by taking~$X$ as the set of binary strings that encodes some class of combinatorial objects, and by taking~$E$ as the pairs of objects that differ in some local changes, for example by a certain Hamming distance or by certain operations on the bitstrings like transpositions, substring shifts, or reversals.

In the following we introduce a simple greedy algorithm to compute a Hamilton path in a binary graph~$G=(X,E)$, where $X\seq \{0,1\}^n$.

Algorithm~G starts at some initial vertex~$\tx\in X$, and then repeatedly moves from the current vertex~$x$ to a neighboring vertex~$y$ that has not been visited before and that minimizes~$\lambda(x,y)$ (recall the definition from Section~\ref{sec:basic-algo}), i.e., the algorithm greedily minimizes the lengths of the modified prefixes.
If all neighbors of~$x$ have been visited, the algorithm terminates.
By definition, the algorithm never visits a vertex twice, i.e., it always computes a path in~$G$.
However, it might terminate before having visited all vertices, i.e., the computed path may not be a Hamilton path.

\begin{algo}{Algorithm~G}{Shortest prefix changes}
This algorithm attempts to greedily compute a Hamilton path in a binary graph~$G=(X,E)$, where $X\seq\{0,1\}^n$, starting from an initial vertex~$\tx$.
\begin{enumerate}[label={\bfseries G\arabic*.}, leftmargin=8mm, noitemsep, topsep=3pt plus 3pt]
\item{} [Initialize] Set $x \gets \tx$.
\item{} [Visit] Visit~$x$.
\item{} [Shortest prefix neighbors] Compute the set~$N$ of unvisited neighbors~$y$ of~$x$ in~$G$ that minimize~$\lambda(x,y)$, i.e., $N \gets \argmin[y \in E(x) \,\wedge\, y \text{ unvisited}\mid \lambda(x,y)]$.
Terminate if~$N=\emptyset$.
\item{} [Tiebreaker+update~$x$] Pick an arbitrary vertex $y\in N$, set $x\gets y$ and goto~G2.
\end{enumerate}
\end{algo}

If the set~$N$ of unvisited neighbors encountered in step~G4 contains more than one element, then we have freedom to pick an arbitrary vertex~$y\in N$.
A tiebreaking rule may involve lexicographic comparisons between vertices in~$N$, or their Hamming distance from~$x$.

Note that Algorithm~G operates completely locally based on the current neighborhood and never uses global information.
Also note that a naive implementation of Algorithm~G may require exponential space to maintain the list of previously visited vertices, in order to be able to compute the set~$N$ in step~G3.
We will show that in many interesting cases, we can make the algorithm history-free, i.e., by introducing suitable additional data structures, we can entirely avoid storing any previously visited vertices.

\section{Genlex order}
\label{sec:genlex}

In this section, we provide a simple sufficient condition for when Algorithm~G succeeds in computing a Hamilton path in a binary graph.
Specifically, these are binary graphs that admit a Hamilton path satisfying a certain ordering property.
We also establish important optimality properties for such orderings.

Let $X\seq\{0,1\}^n$.
An \defi{ordering} of~$X$ is a sequence~$L$ that contains every element of~$X$ exactly once.
Furthermore, an ordering~$L$ of~$X$ is called \defi{genlex}, if all bitstrings with the same suffix appear consecutively in~$L$.
An equivalent recursive definition of genlex order is as follows:
The ordering $L$ is genlex, if $L=L^0,L^1$ or $L=L^1,L^0$, and the sequences~$L^{0-}$ and~$L^{1-}$ are in genlex order.

Colexicographic order is a special case of genlex order in which only the case $L=L^0,L^1$ occurs, i.e., all strings with suffix~$0x$ appear before all strings with suffix~$1x$, for every $x\in\{0,1\}^k$.
In the literature, genlex orderings are sometimes referred to as \emph{suffix-partitioned}~\cite{MR2062208}.
Our notion of genlex ordering is with respect to suffixes, but of course it could also be defined with respect to prefixes instead, generalizing lexicographic order, and such an ordering is sometimes called \emph{prefix-partitioned} in the literature.
Unlike lexicographic order or colexicographic order, which is unique for a given~$X$, there can be many distinct genlex orderings for~$X$.

We associate an ordering~$L=x_1,\ldots,x_\ell$ of~$X$ with a cost, defined by
\begin{equation}
\label{eq:cost}
c(L):=\sum_{i=1,\ldots,\ell-1}\lambda(x_i,x_{i+1}),
\end{equation}
which is the sum of lengths of prefixes that get modified in the ordering.
Genlex orderings are characterized by the property that they minimize this cost among all possible orderings.

\begin{lemma}
\label{lem:genlex-optimal}
Let $L$ and~$L'$ be two orderings of~$X\seq\{0,1\}^n$.
If $L$ is genlex, then we have $c(L)\leq c(L')$.
This inequality is strict if $L'$ is not genlex.
\end{lemma}

Lemma~\ref{lem:genlex-optimal} implies in particular that all genlex orderings of~$X$ have the same cost.

\begin{proof}
Let $L'$ be an arbitrary ordering that is not genlex, i.e., there is a suffix~$bx$ with $b\in\{0,1\}$ such that the strings with suffix~$bx$ do not appear consecutively in~$L'$.
We first consider the case that~$x=\varepsilon$, i.e., in~$L'$ the last bit changes at least twice, w.l.o.g.\ from~$0$ to~$1$ and then back to~$0$ (and possibly again after that).
We thus have $L'=A0,B1,C0,D$, where the sequences $A,B,C$ are nonempty, but $D$ may be empty.
Furthermore, we choose the subsequence~$C$ maximally, i.e., if $D$ is nonempty, then its first bitstring ends with~1.
Consider the sequence $L:=A0,C0,B1,D$, obtained from~$L'$ by swapping the order of the blocks $B1$ and~$C0$.
If $D$ is empty, then $c(L)<c(L')$, as the transition to the first bitstring of~$C0$ costs~$n$ in~$L'$ but strictly less in~$L$.
If $D$ is nonempty, then we also have $c(L)<c(L')$, as the transition to the first bitstring of~$D$ costs~$n$ in~$L'$ but strictly less in~$L$.
We can repeatedly apply such exchange operations to reduce the total cost until the last bit changes at most once.
It remains to consider the case that $x\neq\varepsilon$, in which case the exchange operations discussed before are applied to the subsequence of~$L'$ with the common suffix~$x$, and they are performed until the resulting ordering is genlex.

Note that if $L$ is a genlex ordering, then $c(L)=c(L^{0-})+c(L^{1-})+n$, where the $+n$ comes from the transition between the two blocks~$L^0$ and~$L^1$.
By induction on~$n$ this implies that all genlex orderings have the same cost.

Combining these two arguments proves the lemma.
\end{proof}

The next two lemmas capture further important properties of genlex orderings.

\begin{lemma}
\label{lem:genlex-lambda}
Let $L=x_1,\ldots, x_\ell$ be a genlex ordering of~$X\seq\{0,1\}^n$.
For every two indices $i,j\in[\ell]$ with $i<j$ we have $\lambda(x_i,x_{i+1})\leq \lambda(x_i,x_j)$.
\end{lemma}

\begin{proof}
For the sake of contradiction, suppose that there are indices $i,j \in [\ell]$ with $i<j$ such that $\lambda(x_i,x_{i+1}) > \lambda(x_i,x_j)$.
Thus, the longest common suffix~$s$ of~$x_i$ and~$x_j$ has length $n-\lambda(x_i,x_j)$.
Similarly, the longest common suffix of~$x_i$ and~$x_{i+1}$ has length $n-\lambda(x_i,x_{i+1})<n-\lambda(x_i,x_j)$.
In particular, $s$ is not a suffix of~$x_{i+1}$.
Consequently, the bitstrings with suffix~$s$ do not appear consecutively in~$L$, a contradiction.
\end{proof}

\begin{lemma}
\label{lem:genlex-before-after}
Let $L=x_1,\ldots, x_\ell$ be a genlex ordering of~$X\seq\{0,1\}^n$.
For every three indices $i,j,k\in[\ell]$ with $i<j<k$ we have $\lambda(x_i,x_j)\neq \lambda(x_j,x_k)$.
\end{lemma}

\begin{proof}
For the sake of contradiction, suppose that there are indices $i,j,k\in [\ell]$ with $i<j<k$ such that $\lambda(x_i,x_j)=\lambda(x_j,x_k)$.
Thus, $x_i$, $x_j$ and $x_k$ have a common suffix~$s$ of length~$n-\lambda(x_i,x_j)=n-\lambda(x_j,x_k)$.
Furthermore, $x_i$ and~$x_j$ disagree in the bit to the left of this suffix (at position $\lambda(x_i,x_j)-1=\lambda(x_j,x_k)-1$), and $x_j$ and~$x_k$ disagree in this bit, which implies that $x_i$ and~$x_k$ agree in this bit~$b$.
But this means that the bitstrings with suffix~$bs$ do not appear consecutively in~$L$, a contradiction.
\end{proof}

We now establish that Algorithm~G presented in the previous section generates every genlex Hamilton path if equipped with a suitable tiebreaking rule.
This provides a simple sufficient condition for when the algorithm succeeds in computing a Hamilton path in a given binary graph.
Note that an ordering~$L=x_1,\ldots,x_\ell$ of~$X\seq\{0,1\}^n$ defines a tiebreaking rule~$\tau_L$ that for every given set $N\seq X$ returns the element from~$N$ that appears first in~$L$.
Formally, if $N$ has size~$s$ then there are indices $i_1<\cdots<i_s$ such that $N=\{x_{i_1},\ldots,x_{i_s}\}$, and we define $\tau_L(N):=x_{i_1}$.

\begin{theorem}
\label{thm:algoG-genlex}
Let $G=(X,E)$ be a binary graph that admits a genlex Hamilton path~$L=x_1,\ldots,x_\ell$.
Then Algorithm~\upright{G} with tiebreaking rule~$\tau_L$ and initial vertex $\tx:=x_1$ computes the path~$L$.
\end{theorem}

\begin{proof}
We prove by induction on~$i$ that Algorithm~G visits the vertex~$x_i$ in the $i$th iteration for all $i=1,\ldots,\ell$.
For $i=1$ this is clear by the initialization with $\tx:=x_1$.
For the induction step let~$i>1$.
The vertex~$x_i$ is computed in lines~G3 and~G4 at the end of the $(i-1)$st iteration of the algorithm.
At this point, the visited vertices are $x_1,\dots,x_{i-1}$ in this order, i.e., $x_i$ is unvisited.
Furthermore, $x_i$ is a neighbor of~$x_{i-1}$ in~$G$ by the assumption that $L$ is a Hamilton path, i.e., $x_i\in E(x_{i-1})$.
Furthermore, by Lemma~\ref{lem:genlex-lambda}, $x_i$ minimizes $\lambda(x_{i-1},y)$ among all unvisited neighbors~$y$ of~$x_{i-1}$.
Consequently, we have $x_i\in N$, i.e., $x_i$ is contained in the set~$N$ computed in step~G3.
As $N$ contains only unvisited vertices, the tiebreaking rule~$\tau_L$ indeed selects $x\gets x_i$ as the next vertex in step~G4.
\end{proof}

\begin{remark}
\label{rem:greedy}
Williams~\cite{MR3126386} pioneered the greedy method as a general paradigm to reinterpret known Gray codes and to derive new ones.
Specifically, in his paper he found greedy interpretations of the classical binary reflected Gray code, the Steinhaus-Johnson-Trotter listing of permutations by adjacent transpositions, the Ord-Smith/Zaks~\cite{DBLP:journals/cacm/Ord-Smith67,MR753548} ordering of permutations by prefix reversals, and for the rotation Gray code for binary trees due to Lucas, Roelants van Baronaigien, and Ruskey~\cite{MR1239499}.

The greedy method has also been very useful in discovering new Gray codes for (generalized) permutations~\cite{MR3513761,MR4346446}, for spanning trees of special graphs~\cite{DBLP:conf/cocoon/CameronGS21}, and for spanning trees of arbitrary graphs and more generally bases of a matroid~\cite{MR4473269}.
Also, the permutation-based framework for combinatorial generation proposed by Hartung, Hoang, M\"utze and Williams~\cite{MR4391718} relies on a simple greedy algorithm.

Theorem~\ref{thm:algoG-genlex} now provides us with an explanation for the success of the greedy method.
Furthermore, the recent survey~\cite{mutze:2023} lists a large number of Gray codes from the literature that have the genlex property (there are more than 50 occurrences of the word `genlex'), and all of those can now be interpreted as the result of a simple algorithm that greedily minimizes the lengths of modified prefixes in each step.
\end{remark}

\begin{remark}
Admittedly, one can consider the proof of Theorem~\ref{thm:algoG-genlex} as `cheating', as it builds knowledge about~$L$ into the tiebreaking rule~$\tau_L$ of the algorithm, knowledge that one typically does not have when the goal is to come up with an algorithm to produce a listing~$L$ in the first place.
Consequently, in practice the challenge is to come up with a tiebreaking rule that uses only local information about the neighborhood of the current vertex~$x$ and that can be computed efficiently.
In some cases, discussed in the next section, we are in the even better position that Algorithm~G works for \emph{every} choice of tiebreaking rule, which gives dramatic additional flexibility.
\end{remark}

\section{Prefix graphs and history-free implementation}
\label{sec:prefix}

In this section, we exhibit a class of binary graphs for which Algorithm~G succeeds for \emph{every} tiebreaking rule, and for \emph{every} choice of initial vertex~$\tx$.
Furthermore, for these binary graphs we can implement Algorithm~G in a history-free way, i.e., without maintaining the list of all previously visited vertices, by maintaining some simple additional data structures.
Furthermore, we will formulate two auxiliary problems that, when solved efficiently, imply that Algorithm~G runs with short delay between consecutively visited vertices.
In later sections, we will solve these auxiliary problems by combinatorial optimization methods.

\subsection{Prefix graphs}
\label{sec:prefix-graphs}

For a binary graph~$G=(X,E)$, where $X\seq \{0,1\}^n$, and for $b\in\{0,1\}$ we define $G^{b-}:=(X^{b-},E^{b-})$ with $E^{b-}:=\{(x,y) \mid x,y\in X^{b-} \text{ and } (xb,yb) \in E\}$.
We say that~$G=(X,E)$ is a \defi{prefix graph} if $X=\emptyset$, or $n=0$ and $X=\{\varepsilon\}$, or $n>0$ and the following two conditions hold:
\begin{enumerate}
\item[(p1)] $G^{0-}$ and $G^{1-}$ are prefix graphs;
\item[(p2)] If $X^0$ and $X^1$ are both nonempty, then for every $b\in\{0,1\}$ and for every vertex~$x \in X^b$ there exists a vertex~$y \in X^{\ol{b}}$ such that $(x,y) \in E$.
\end{enumerate}

Even though the condition~(p2) may seem rather restrictive, we will see that many interesting binary graphs are indeed prefix graphs.
In particular, the skeleton of \emph{every} 0/1-polytope is a prefix graph; see Lemma~\ref{lem:01prefix} below.

\begin{lemma}
\label{lem:prefix-genlex}
For every prefix graph~$G=(X,E)$ and every vertex~$\tx\in X$, there is a genlex Hamilton path in~$G$ starting at~$\tx$.
\end{lemma}

\begin{proof}
We argue by induction on~$n$.
The base case $n=0$ holds trivially.
For the induction step we assume that~$n>0$.
We assume without loss of generality that $\tx$ has last bit~0, as the other case is symmetric.
By this assumption~$X^0$ and~$G^{0-}$ are nonempty.
From condition~(p1) of prefix graphs we obtain that~$G^{0-}$ is a prefix graph, so by induction there is a genlex Hamilton path~$L$ in~$G^{0-}$ starting at~$\tx^-$.
Let $x'$ be the last vertex of~$L$.
If $X^1$ is empty, then $L0$ is a genlex Hamilton path in~$G$ and we are done.
Otherwise, $G^{1-}$ is nonempty.
From condition~(p1) we obtain that~$G^{1-}$ is a prefix graph.
Furthermore, from condition~(p2) we obtain that there is a vertex~$\ty\in X^1$ such that $(x'0,\ty)\in E$.
Therefore, by induction there is a genlex Hamilton path~$M$ in~$G^{1-}$ starting at~$\ty^-$.
The concatenation $L0,M1$ is the desired genlex Hamilton path in~$G$.
\end{proof}

Our next theorem strengthens Lemma~\ref{lem:prefix-genlex} and makes it algorithmic.
This fundamental theorem asserts that Algorithm~G succeeds in computing a Hamilton path on \emph{every} prefix graph, \emph{regardless} of the choice of tiebreaking rule, and \emph{regardless} of the choice of initial vertex.
The importance of this result can hardly be overstated, as it gives us dramatic flexibility in many applications.

\begin{theorem}
\label{thm:algoG-prefix}
Let $G=(X,E)$ be a prefix graph.
For every tiebreaking rule and every initial vertex~$\tx$, Algorithm~\upright{G} computes a genlex Hamilton path on~$G$ starting at~$\tx$.
\end{theorem}

The proof follows the same strategy as the proof of Lemma~\ref{lem:prefix-genlex}. 

\begin{proof}
We argue by induction on~$n$.
The base case~$n=0$ holds trivially.
For the induction step we assume that~$n>0$.
We assume without loss of generality that~$\tx$ has last bit~0, as the other case is symmetric.
By this assumption $X^0$ and $G^{0-}$ are nonempty.
From condition~(p1) of prefix graphs we obtain that $G^{0-}$ is a prefix graph, so by induction Algorithm~G with input $G^{0-}$ computes a genlex Hamilton path~$L$ in~$G^{0-}$ starting at~$\tx^-$.
Let $x'$ be the last vertex of~$L$.
Observe that Algorithm~G with input~$G$ produces the path~$L0$, whose last vertex is~$x'0$, and we now consider the iteration of the algorithm where the vertex~$x'0$ is visited.
After visiting~$x'0$ all vertices of~$X^0$ have been visited.
If $X^1$ is empty, then the set~$N$ computed in step~G3 is empty, the algorithm terminates and we are done.
Otherwise, $G^{1-}$ is nonempty.
From condition~(p1) we obtain that~$G^{1-}$ is a prefix graph.
Furthermore, from condition~(p2) we obtain that there is a vertex~$y\in X^1$ such that $(x'0,y)\in E$, i.e., we have $y\in E(x'0)$.
Furthermore, $y$ is unvisited, as its last bit equals~1 and therefore $y\in N$, implying that $N\neq\emptyset$.
Consequently in step~G4 the algorithm moves to some vertex~$\ty\in N\seq X^1$, which is true regardless of the tiebreaking rule being used.
We know by induction that Algorithm~G with input~$G^{1-}$ computes a genlex Hamilton path~$M$ in~$G^{1-}$ starting at~$\ty^-$.
From this we conclude that Algorithm~G with input~$G$ computes the genlex Hamilton path~$L0,M1$.
\end{proof}

\subsection{Suffix trees and branchings}
\label{sec:suffix-tree}

We now describe how to equip Algorithm~G with additional data structures so that it does not need to store any previously visited vertices.

The key observation is that the suffixes of a genlex ordering have a binary tree structure.
Formally, let $L$ be a genlex ordering of a set~$X\seq\{0,1\}^n$.
The \defi{suffix tree} $\cT(L)$ is an ordered rooted tree whose vertices are all possible suffixes of~$X$, with the following adjacencies; see Figure~\ref{fig:suffix-tree}:
\begin{itemize}[itemsep=0ex,parsep=0ex,leftmargin=2.5ex]
\item the empty suffix~$\varepsilon$ is the root of~$\cT(L)$;
\item for every suffix~$s$ of length $k$, its children in~$\cT(L)$ are the suffixes of length~$k+1$ that have $s$ as a common suffix, and the order of children from left to right corresponds to the order of the suffixes in~$L$.
\end{itemize}
Note that the set of leaves of~$\cT(L)$ equals~$X$, and the sequence of leaves in~$\cT(L)$ from left to right equals~$L$.
Furthermore, every vertex in~$\cT(L)$ has either one or two children.

\begin{figure}
\begin{tabular}{ccc}
\includegraphics[page=1]{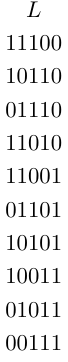} &
\includegraphics[page=2]{tree} &
\includegraphics[page=3]{tree} \\
(a) & (b) & (c)
\end{tabular}
\caption{(a)~A genlex ordering~$L$ of all bitstrings of length~5 with three~1s.
(b)~The binary tree structure in~$L$.
(c)~The corresponding suffix tree~$\cT(L)$.}
\label{fig:suffix-tree}
\end{figure}

When producing a genlex ordering~$L$, Algorithm~G traverses the leaves of the suffix tree~$\cT(L)$.
To obtain a history-free implementation, it is enough to store information about the current leaf and the branchings on the path from that leaf to the root.
Formally, for~$x\in X$ we define
\[B(x):=\{\lambda(x,y) \mid y \in X-x \},\]
and we refer to $B(x)$ as the set of \defi{branchings of~$x$}.
This definition is independent of~$L$ and hence of~$\cT(L)$.
Note however that in a suffix tree~$\cT(L)$ and for a leaf~$x\in X$, we have $i\in B(x)$ if and only if the node in distance~$i$ from~$x$ in~$\cT(L)$ (which has depth $n-i$ in the tree) has two children; see Figure~\ref{fig:branchings}.

For a genlex ordering~$L=x_1,\ldots,x_\ell$ of~$X$ and~$i\in[\ell]$, we define the set of \defi{unseen branchings of~$x_i$ w.r.t.~$L$} as
\begin{equation}
\label{eq:BLxi}
B_L(x_i):=\{ \lambda(x_i,x_j) \mid i<j\leq \ell \}\seq B(x_i).
\end{equation}
These are branchings in~$\cT(L)$ that lead to children that are visited after~$x_i$ in~$L$; see Figure~\ref{fig:branchings}.

\begin{figure}
\includegraphics[page=4]{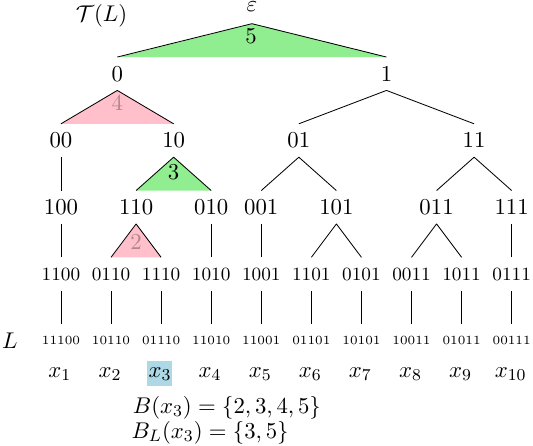}
\caption{Illustration of branchings and unseen branchings for the vertex~$x_3$ in the suffix tree from Figure~\ref{fig:suffix-tree}.}
\label{fig:branchings}
\end{figure}

The next lemma states two important properties about the quantities defined before.

\begin{lemma}
\label{lem:branchings-properties}
Let $L=x_1,\dots,x_\ell$ be a genlex ordering of $X\seq\{0,1\}^n$.
For every $1\leq i<\ell$ the minimum unseen branching $\beta:=\min B_L(x_i)$ satisfies the following properties:
\begin{enumerate}[label=(\roman*),leftmargin=8mm, noitemsep, topsep=1pt plus 1pt]
\item for every $x_j\in X-x_i$ with $\lambda(x_i,x_j)=\beta$ we have $j>i$;
\item we have $B_L(x_i)\setminus \{\beta\}=B_L(x_{i+1})\setminus [\beta-1]$.
\end{enumerate}
\end{lemma}

\begin{proof}
By Lemma~\ref{lem:genlex-lambda} we have $\lambda(x_i,x_{i+1})=\beta$.
To prove~(i), suppose for the sake of contradiction that there is some $x_j\in X-x_i$ with $\lambda(x_i,x_j)=\beta$ and $j<i$.
Then we have $j<i<i+1$ with $\lambda(x_j,x_i)=\lambda(x_i,x_{i+1})=\beta$, contradicting Lemma~\ref{lem:genlex-before-after}.

To prove~(ii), first note that $\beta\notin B_L(x_{i+1})$, as otherwise there would be $j>i+1$ such that $\lambda(x_i,x_{i+1})=\lambda(x_{i+1},x_j)$, contradicting Lemma~\ref{lem:genlex-before-after}.
It remains to argue that all elements strictly larger than~$\beta$ are the same in~$B_L(x_i)$ and~$B_L(x_{i+1})$.
Note that if $\lambda(x_{i+1},x_j)\leq \lambda(x_i,x_{i+1})=\beta$ for all $j>i+1$, then there are no elements strictly larger than~$\beta$ in either of the two sets, so we are done.
Otherwise let $j>i+1$ be an index such that $\lambda(x_{i+1},x_j)>\lambda(x_i,x_{i+1})=\beta$.
Then the longest common suffix of~$x_i$ and~$x_{i+1}$ properly contains the longest common suffix of~$x_{i+1}$ and~$x_j$.
Hence, the longest common suffix of~$x_{i+1}$ and~$x_j$ is also the longest common suffix of~$x_i$ and~$x_j$.
Consequently, we have $\lambda(x_i,x_j)=\lambda(x_{i+1},x_j)$ for every $j>i+1$ such that $\lambda(x_{i+1},x_j)>\lambda(x_i,x_{i+1})=\beta$.
This completes the proof of~(ii).
\end{proof}

\subsection{History-free implementation}

An \defi{interval} is a subset of consecutive natural numbers.
For an interval $I\seq [n]$ we define
\[
\lambda_I(x,y) := \begin{cases}
  \lambda(x,y) & \text{if }\lambda(x,y) \in I, \\
  \infty & \text{otherwise}.
\end{cases}
\]
Note that we have $\lambda_{[n]}(x,y)=\lambda(x,y)$.

To make Algorithm~G history-free, we need to get rid of the qualification `$y$ unvisited' in the computation of the set~$N$ in line~G3.
This is achieved in Algorithm~G\sss{} stated below, which takes as input a binary graph~$G=(X,E)$, where $X\seq\{0,1\}^n$.
The algorithm keeps track of the unseen branchings of the current vertex~$x$ by maintaining a stack~$U$ of disjoint intervals that cover all unseen branchings of~$x$, with the property that each interval~$I$ on the stack contains at least one unseen branching of~$x$.
The intervals appear on the stack~$U$ in increasing order (of the numbers in each interval) starting at the top of the stack and ending at the bottom of the stack, i.e., the interval with the smallest numbers is at the top, and the interval with the largest numbers is at the bottom.
For each interval~$I$ on the stack, the variable~$\beta_I$ stores the minimum unseen branching in~$I$, i.e., we have $\beta_I=\min \{\lambda_I(x,y)\mid y\in X-x\}$.
Note here Lemma~\ref{lem:branchings-properties}~(i), so no extra qualification `$y$ unvisited' is needed in this minimization.
There might be more than one unseen branching in~$I$, but only the minimum one is stored in~$\beta_I$.

\begin{algo}{Algorithm~G\sss{}}{History-free shortest prefix changes}
This algorithm attempts to greedily compute a Hamilton path in a binary graph~$G=(X,E)$, where $X\seq \{0,1\}^n$, starting from an initial vertex~$\tx$.
\begin{enumerate}[label={\bfseries G\arabic*.}, leftmargin=8mm, noitemsep, topsep=3pt plus 3pt]
\item{} [Initialize] Set $x \gets \tx$ and call $\branching([n])$.
\item{} [Visit] Visit~$x$.
\item{} [Min.\ unseen branching] Terminate if $U$ is empty.
Otherwise set $I\!\gets\! U.\pop()$ and $\beta\gets\beta_I$.
\item{} [Shortest prefix neighbors] Compute the set~$N$ of neighbors~$y$ of~$x$ in~$G$ with $\lambda(x,y)=\beta$, i.e., $N\gets \{y\in E(x)\mid \lambda(x,y)=\beta\}$.
\item{} [Tiebreaker+update~$x$] Pick an arbitrary vertex $y \in N$ and set $x\gets y$.
\item{} [Update~$U$] Call $\branching(I\setminus [\beta])$ and $\branching([\beta-1])$, and goto~G2.
\end{enumerate}
\end{algo}

Algorithm~G\sss{} calls the following auxiliary function to update the stack~$U$ for a given interval~$I$.
This function reads the current vertex~$x$ and if the given interval~$I$ contains an unseen branching it modifies the variables~$\beta_I$ and~$U$.

\vspace{1.5mm}
\hspace{-6.7mm}
\mybox{$\branching(I)$:
Compute $\beta\gets\min \{\lambda_I(x,y)\mid y\in X-x\}$.
If $\beta<\infty$ set $\beta_I\gets \beta$ and $U.\push(I)$.}
\vspace{-2mm}

\begin{figure}
\includegraphics[page=5]{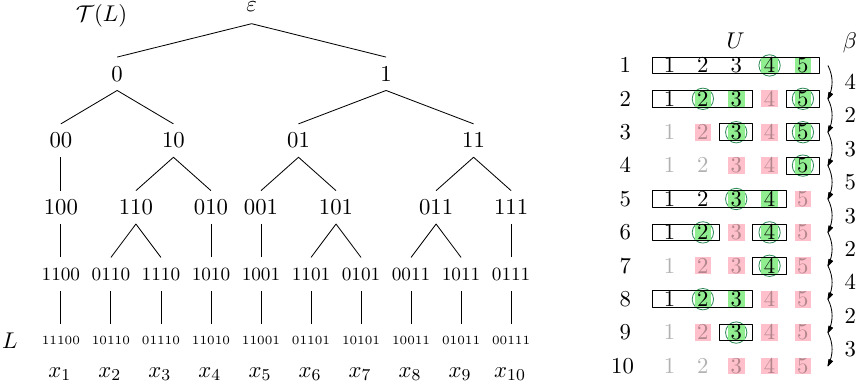}
\caption{The state of the variables~$U$ and~$\beta$ of Algorithm~G\sss{} while traversing the genlex tree from Figure~\ref{fig:suffix-tree}.
Each interval~$I$ on the stack~$U$ is indicated by a box.
Every branching is highlighted by a shaded square, with unseen branchings printed black and the others grayed out.
Furthermore, the minimum unseen branching in each interval is circled.}
\label{fig:algoGs}
\end{figure}

The stack~$U$ is initialized in step~G1, the interval~$I$ containing the minimum unseen branching~$\beta_I$ of the current vertex~$x$ is retrieved from the stack in step~G3, and the stack is updated in step~G6.
By Lemma~\ref{lem:branchings-properties}~(i)+(ii), these updates preserve the invariants about the intervals~$I$ on the stack~$U$ and the associated numbers~$\beta_I$ stated before the algorithm.

It remains to argue that the set~$N$ of neighbors computed in step~G3 of Algorithm~G is the same as the set~$N$ computed in step~G4 of Algorithm~G\sss{}.
Indeed, we have
\begin{equation}
\label{eq:N-comp}
\begin{split}
N&=\argmin[y \in E(x) \,\wedge\, y \text{ unvisited}\mid \lambda(x,y)]  \\
 &=\{y\in E(x) \,\wedge\, y \text{ unvisited} \mid \lambda(x,y)=\beta \} \\
 &=\{y\in E(x) \mid \lambda(x,y)=\beta \},
\end{split}
\end{equation}
where the quantity~$\beta$ is the minimum unseen branching of the current vertex~$x$ defined in line~G3 of Algorithm~G\sss{}, and in the last step we use Lemma~\ref{lem:branchings-properties}~(i).

Summarizing these observations we obtain the following result.

\begin{theorem}
\label{thm:algoGs-genlex}
Let $G=(X,E)$ be a binary graph that admits a genlex Hamilton path~$L=x_1,\ldots,x_\ell$.
Then Algorithm~\upright{G\sss{}} with tiebreaking rule~$\tau_L$ and initial vertex $\tx:=x_1$ produces the same output as Algorithm~\upright{G}, namely the path~$L$.
\end{theorem}

Combining this result with Theorem~\ref{thm:algoG-prefix}, we obtain the following.

\begin{theorem}
\label{thm:algoGs-prefix}
Let $G=(X,E)$ be a prefix graph.
For every tiebreaking rule and every initial vertex~$\tx$, Algorithm~\upright{G\sss{}} computes a genlex Hamilton path on~$G$ starting at~$\tx$.
\end{theorem}

\subsection{Two auxiliary problems}
\label{sec:probAB}

From the pseudocode of Algorithm~G\sss{} we can extract the following two computational problems, which, if solved efficiently, directly lead to an efficient algorithm for computing a Hamilton path in a prefix graph:
\begin{enumerate}[label=\mybox{\Alph*},leftmargin=8mm, noitemsep, topsep=1pt plus 1pt]
\item Given a set $X\seq\{0,1\}^n$, an element~$x \in X$ and an interval~$I\seq [n]$, compute $\min\{\lambda_I(x,y)\mid y \in X-x\}$.
\item Given a binary graph~$G=(X,E)$, where $X\seq\{0,1\}^n$, an element~$x \in X$ and an integer~$\beta \in [n]$ with $N:=\{y\in E(x)\mid \lambda(x,y)=\beta\}\neq \emptyset$, compute an element in~$N$.
\end{enumerate}

We note that in both problems, the set~$X$ of bitstrings may be given implicitly via some other parameter; recall Table~\ref{tab:01poly}.

\begin{theorem}
\label{thm:algoGs-time}
Let $G=(X,E)$ be a prefix graph, and suppose that problems~\upright{A} and~\upright{B} can be solved in time~$t_{\upright{A}}$ and~$t_{\upright{B}}$, respectively.
Then for every tiebreaking rule and every initial vertex~$\tx$, Algorithm~\upright{G\sss{}} computes a genlex Hamilton path on~$G$ starting at~$\tx$ with delay~$\cO(t_{\upright{A}}+t_{\upright{B}})$.
\end{theorem}

The initialization time of Algorithm~G\sss{} is $\cO(t_{\upright{A}})$, which is majorized by the delay~$\cO(t_{\upright{A}}+t_{\upright{B}})$, and the required space is the sum of the space needed to solve problems~A and~B.

\begin{proof}
This is immediate from Theorem~\ref{thm:algoGs-prefix} and the fact that each iteration of Algorithm~G\sss{} consists of solving constantly many instances of problems~A and~B.
\end{proof}

\section{A bridge to combinatorial optimization}
\label{sec:opt}

In this section, we consider 0/1-polytopes, i.e., polytopes that arise as the convex hull~$\conv(X)$ of a set of binary strings~$X\seq\{0,1\}^n$.
We first show that the skeleton of a 0/1-polytope is a prefix graph.
Furthermore, we show that problems~A and~B reduce to solving a particular linear optimization problem on the polytope (recall Section~\ref{sec:reduction}).
Consequently, if we can solve this optimization problem efficiently, then we obtain an efficient algorithm for computing a Hamilton path on the skeleton of the polytope.

\subsection{Skeleta of 0/1-polytopes are prefix graphs}

Recall the definition of Hamming distance~$d(x,y)$ from Section~\ref{sec:basic-algo}.

\begin{lemma}[{\cite[Proposition~2.3]{MR638286}}]
\label{lem:hamming-edge}
Let $P$ be a 0/1-polytope and let $G=(X,E)$ with $X\seq\{0,1\}^n$ be its skeleton.
Suppose that $X^0$ and $X^1$ are both nonempty, and let $x$ be a vertex of~$X^b$ for some $b\in\{0,1\}$.
If a vertex~$y \in X^{\ol{b}}$ minimizes $d(x,y)$, then we have $(x,y)\in E$.
\end{lemma}

We apply Lemma~\ref{lem:hamming-edge} to prove that the skeleton of a 0/1-polytope is a prefix graph.

\begin{lemma}
\label{lem:01prefix}
Let $P$ be a 0/1-polytope and let $G$ be its skeleton.
Then $G$ is a prefix graph.
\end{lemma}

\begin{proof}
Let $G=:(X,E)$ with $X\seq\{0,1\}^n$, i.e., $P=\conv(X)$.
We argue by induction on~$n$.
The induction basis~$n=0$ is trivial.
For the induction step we assume that~$n>0$.
For $b\in\{0,1\}$ we define $P^{b-}:=\conv(X^{b-})$, and we observe that $G^{b-}$ is the skeleton of~$P^{b-}$.
By induction, we obtain that $G^{0-}$ and~$G^{1-}$ are prefix graphs, so condition~(p1) in the definition given in Section~\ref{sec:prefix-graphs} is satisfied.
It remains to prove property~(p2), under the assumption that~$X^0$ and~$X^1$ are both nonempty.
Let $b\in\{0,1\}$, consider an arbitrary vertex~$x\in X^b$, and let $y \in X^{\ol{b}}$ be such that~$d(x,y)$ is minimized.
Then Lemma~\ref{lem:hamming-edge} shows that $(x,y) \in E$, and the lemma follows.
\end{proof}

\subsection{Proof of Theorem~\ref{thm:algoP-poly} and a history-free version of Algorithm~P}
\label{sec:algoP}

We obtain that Algorithm~G and Algorithm~G\sss{} compute a Hamilton path on the skeleton of \emph{every} 0/1-polytope, \emph{regardless} of the choice of tiebreaking rule, and \emph{regardless} of the choice of initial vertex.

\begin{theorem}
\label{thm:algoG-poly}
Let $P$ be a 0/1-polytope and let $G$ be its skeleton.
For every tiebreaking rule and every initial vertex~$\tx$, Algorithm~\upright{G} computes a genlex Hamilton path on~$G$ starting at~$\tx$.
\end{theorem}

\begin{proof}
Combine Lemma~\ref{lem:01prefix} with Theorem~\ref{thm:algoG-prefix}.
\end{proof}

\begin{theorem}
\label{thm:algoGs-poly}
Let $P$ be a 0/1-polytope and let $G$ be its skeleton.
For every tiebreaking rule and every initial vertex~$\tx$, Algorithm~\upright{G\sss{}} computes a genlex Hamilton path on~$G$ starting at~$\tx$.
\end{theorem}

\begin{proof}
Combine Lemma~\ref{lem:01prefix} with Theorem~\ref{thm:algoGs-prefix}.
\end{proof}

We now specialize Algorithms~G and~G\sss{} further, and remove any references to the skeleton of the polytope with the help of Lemma~\ref{lem:hamming-edge}.
Applying this specialization to Theorem~\ref{thm:algoG-poly} yields Theorem~\ref{thm:algoP-poly} stated in Section~\ref{sec:basic-algo}.

\begin{proof}[Proof of Theorem~\ref{thm:algoP-poly}]
By slightly modifying line~G3 in Algorithm~G, we obtain Algorithm~P stated in Section~\ref{sec:basic-algo}.
Specifically, we may split the computation of the set~$N$ in line~G3 into two steps, namely $\beta\leftarrow\min \{ \lambda(x,y) \mid y\in E(x)\,\wedge\,y \text{ unvisited} \}$ and $N\leftarrow \{y\in E(x)\mid \lambda(x,y)=\beta\}$ (recall~\eqref{eq:N-comp}).
We then have
\[ \beta = \min \{ \lambda(x,y) \mid y\in E(x)\,\wedge\,y \text{ unvisited} \}
= \min \{\lambda(x,y) \mid y\in X-x \,\wedge\, y \text{ unvisited} \} \]
by Lemma~\ref{lem:genlex-lambda}.
Furthermore, by Lemma~\ref{lem:hamming-edge}, every vertex with minimum Hamming distance from the current vertex~$x$ is a neighbor on the skeleton, i.e., we have
\begin{equation}
\label{eq:E-X-nonempty}
\{y\in E(x)\mid \lambda(x,y)=\beta\}\;\supseteq\; \argmin[y\in X-x \,\wedge\, \lambda(x,y)=\beta \mid d(x,y)]\neq \emptyset.
\end{equation}
These two modifications yield the steps~P3 and~P4 in Algorithm~P.
The theorem now follows from~Theorem~\ref{thm:algoG-poly}.
\end{proof}

The observation~\eqref{eq:E-X-nonempty} also allows us to specialize the history-free Algorithm~G\sss{} to a history-free algorithm on 0/1-polytopes, by slightly modifying line~G4, yielding Algorithm~P\sss{} stated below.

\begin{algo}{Algorithm~P\sss{}}{History-free traversal of 0/1-polytope by shortest prefix changes}
For a set $X\seq \{0,1\}^n$, this algorithm greedily computes a Hamilton path on the skeleton of the 0/1-polytope $\conv(X)$, starting from an initial vertex~$\tx$.
\begin{enumerate}[label={\bfseries P\arabic*.}, leftmargin=8mm, noitemsep, topsep=3pt plus 3pt]
\item{} [Initialize] Set $x \gets \tx$ and call $\branching([n])$.
\item{} [Visit] Visit~$x$.
\item{} [Min.\ unseen branching] Terminate if $U$ is empty.
Otherwise set $I\!\gets\! U.\pop()$ and $\beta\gets\beta_I$.
\item{} [Closest vertices] Compute the set~$N$ of vertices~$y$ with $\lambda(x,y)=\beta$ of minimum Hamming distance from~$x$, i.e., $N\gets \argmin[y\in X-x \,\wedge\, \lambda(x,y)=\beta \mid d(x,y)]$.
\item{} [Tiebreaker+update~$x$] Pick an arbitrary vertex $y \in N$ and set $x\gets y$.
\item{} [Update~$U$] Call $\branching(I\setminus [\beta])$ and $\branching([\beta-1])$, and goto~P2.
\end{enumerate}
\end{algo}

The corresponding specialization of Theorem~\ref{thm:algoGs-poly} reads as follows.

\begin{theorem}
\label{thm:algoPs-poly}
Let $X\seq\{0,1\}^n$.
For every tiebreaking rule and every initial vertex~$\tx$, Algorithm~\upright{P\sss{}} computes a genlex Hamilton path on the skeleton of~$\conv(X)$ starting at~$\tx$.
\end{theorem}

Furthermore, the auxiliary problem~B introduced in Section~\ref{sec:probAB} can be specialized for Algorithm~P\sss{} as follows:
\begin{enumerate}[label=\mybox{C},leftmargin=8mm, noitemsep, topsep=1pt plus 1pt]
\item Given a set $X\seq\{0,1\}^n$, an element~$x \in X$ and an integer~$\beta \in [n]$ with
$N:=\argmin[y\in X-x \,\wedge\, \lambda(x,y)=\beta \mid d(x,y)]\neq\emptyset,$
compute an element in~$N$.
\end{enumerate}

We thus obtain the following specialization of Theorem~\ref{thm:algoGs-time}.

\begin{theorem}
\label{thm:algoP-time}
Let $X\seq\{0,1\}^n$, and suppose that problems~\upright{A} and~\upright{C} can be solved in time~$t_{\upright{A}}$ and~$t_{\upright{C}}$, respectively.
Then for every tiebreaking rule and every initial vertex~$\tx$, Algorithm~\upright{P\sss{}} computes a genlex Hamilton path on the skeleton of~$\conv(X)$ starting at~$\tx$ with delay~$\cO(t_{\upright{A}}+t_{\upright{C}})$.
\end{theorem}

The initialization time of Algorithm~P\sss{} is $\cO(t_{\upright{A}})$, which is majorized by the delay~$\cO(t_{\upright{A}}+t_{\upright{C}})$, and the required space is the sum of the space needed to solve problems~A and~C.

\subsection{Reducing problems~A and~C to a single linear optimization problem}
\label{sec:PO}

It turns out that both problems~A and~C can be reduced to one or more instances of the following optimization problem, referred to as \defi{linear optimization with prescription}.
\begin{enumerate}[label=\mybox{LOP},leftmargin=12mm, noitemsep, topsep=1pt plus 1pt]
\item Given a set~$X\seq\{0,1\}^n$, a weight vector~$w\in W^n$ with $W\seq\mathbb{R}$, and disjoint sets $P_0,P_1\seq [n]$, compute an element in
$N:=\argmin[y\in X\,\wedge\, y_{P_0}=0 \,\wedge\, y_{P_1}=1 \mid w\cdot y]$,
or decide that this problem is infeasible, i.e., $N=\emptyset$.
\end{enumerate}
Recall that $y_{P_b}=b$ for $b\in\{0,1\}$ is a shorthand notation for $y_i=b$ for all $i\in P_b$.
We refer to~$W$ as the \defi{weight set}.

We first show that problem~A can be solved by combining an algorithm for problem~LOP with binary search.

\begin{lemma}
\label{lem:algoA-LOP}
Suppose that problem~\upright{LOP} with weight set~$W=\{-1,0,+1\}$ can be solved in time~$t_{\upright{LOP}}=\Omega(n)$.
Then problem~\upright{A} can be solved in time~$\cO(t_{\upright{LOP}}\log n)$.
\end{lemma}

\begin{proof}
Consider a set~$X\in\{0,1\}^n$, an element~$x\in X$ and an interval~$I\in[n]$ as input for problem~A.
First note that $\min \{\lambda_I(x,y)\mid y\in X-x\} \seq [n]\cup \infty$ and that
\begin{equation}
\label{eq:mono}
\min \{\lambda_I(x,y)\mid y\in X-x\} \leq \alpha
\end{equation}
is a monotone property in $\alpha\in I$.
Therefore, it is enough to show that~\eqref{eq:mono} can be decided in time~$\cO(t_{\upright{LOP}})$, as then we can compute $\min \{\lambda_I(x,y)\mid y\in X-x\}$ in time $\cO(t_{\upright{LOP}}\log n)$ by doing binary search.
For the given integer~$\alpha\in I$ we define
\begin{subequations}
\label{eq:PwA}
\begin{equation}
\label{eq:wiA}
w_i:=\begin{cases}
-1 & \text{if } i\geq \min I\text{ and } x_i=0, \\
+1 & \text{if } i\geq \min I\text{ and } x_i=1, \\
0  & \text{if } i<\min I,
\end{cases}
\end{equation}
and
\begin{equation}
P_b:=\{i>\alpha\mid x_i=b\} \text{ for } b\in\{0,1\}.
\end{equation}
\end{subequations}
We claim that
\[\mu:=\min\limits_{y\in X\,\wedge\,y_{P_0}=0\,\wedge\,y_{P_1}=1} w\cdot y<w\cdot x=:a\]
if and only if~\eqref{eq:mono} holds.
Indeed, if $\mu<a$, then there is a~$y^*\in X$ with $y^*_{P_0}=0$, $y^*_{P_1}=1$, and~$y^*_i\neq x_i$ for some $i\in I$ with $i\leq \alpha$.
It follows that $\lambda_I(x,y^*)\leq \alpha$, which implies~\eqref{eq:mono}.
Conversely, if~\eqref{eq:mono} holds, then there is $y^*\in X-x$ with $\lambda_I(x,y^*)\leq \alpha$, i.e., there is a position $i\in I$ with $i\leq\alpha$ such that $y^*_i\neq x_i$ and $y^*_j=x_j$ for all $j\geq i+1$, in particular $y^*_{P_0}=0$ and $y^*_{P_1}=1$.
As $y^*_i\neq x_i$ we have $w\cdot y^*<a$ and therefore~$\mu<a$.
This completes the proof of the lemma.
\end{proof}

We now show that problem~C can also be reduced to problem~LOP.

\begin{lemma}
\label{lem:algoC-LOP}
Suppose that problem~\upright{LOP} with weight set $W=\{-1,+1\}$ can be solved in time~$t_{\upright{LOP}}=\Omega(n)$.
Then problem~\upright{C} can be solved in time~$\cO(t_{\upright{LOP}})$.
\end{lemma}

\begin{proof}
Consider a set $X\seq\{0,1\}^n$, an element $x\in X$ and an integer $\beta\in[n]$ with
$N:=\argmin[y\in X-x \,\wedge\, \lambda(x,y)=\beta \mid d(x,y)]\neq\emptyset$
as input for problem~C.
We define
\begin{subequations}
\label{eq:PwC}
\begin{equation}
\label{eq:wiC}
w_i:=\begin{cases}
+1 & \text{if } x_i=0, \\
-1 & \text{if } x_i=1,
\end{cases}
\end{equation}
and
\begin{equation}
P_b:=\{\beta\mid x_\beta=\ol{b}\}\cup \{i>\beta\mid x_i=b\} \text{ for } b\in\{0,1\}.
\end{equation}
\end{subequations}
By this definition we have $d(x,y)=w\cdot (y-x)$, and as $x$ is fixed, minimization of $d(x,y)$ is the same as minimization of~$w\cdot y$.
Consequently, we have
\[
N=\argmin[y\in X-x\,\wedge\,\lambda(x,y)=\beta \mid d(x,y)]=
\argmin[y\in X\,\wedge\,y_{P_0}=0\,\wedge\,y_{P_1}=1 \mid w\cdot y].
\]
In this calculation we used that from our definition of~$P_0$ and~$P_1$, the conditions $y_{P_0}=0$ and $y_{P_1}=1$ are equivalent to $y_\beta=\ol{x_\beta}$ and $y_i=x_i$ for $i>\beta$, which are equivalent to $y\neq x$ and $\lambda(x,y)=\beta$.
This completes the proof of the lemma.
\end{proof}

Note the opposite signs of the weights in~\eqref{eq:wiA} and~\eqref{eq:wiC} w.r.t.~$x$.
This is because the first minimization problem rewards~$y$ to differ from~$x$ as much as possible, whereas the second problem rewards~$y$ to agree with~$x$ as much as possible.

Combining these reductions yields the following fundamental result, which says that efficiently solving prescription optimization on~$X$ yields an efficient algorithm for computing a Hamilton path on the skeleton of~$\conv(X)$.

\begin{theorem}
\label{thm:algoPLOP-time}
Let $X\seq\{0,1\}^n$ and suppose that problem~\upright{LOP} with weight set $W=\{-1,0,+1\}$ can be solved in time~$t_{\upright{LOP}}=\Omega(n)$.
Then for every tiebreaking rule and every initial vertex~$\tx$, Algorithm~\upright{P\sss{}} computes a genlex Hamilton path on the skeleton of~$\conv(X)$ starting at~$\tx$ with delay~$\cO(t_{\upright{LOP}}\log n)$.
\end{theorem}

The initialization time of Algorithm~P\sss{} is the same as the delay, and the required space is the same as the space needed to solve problem~LOP.

\begin{proof}
Combine Lemmas~\ref{lem:algoA-LOP} and~\ref{lem:algoC-LOP} with Theorem~\ref{thm:algoP-time}.
\end{proof}

\subsection{Cost-optimal solutions}
\label{sec:cost-optimal}

In this section, we provide a variant of Theorem~\ref{thm:algoPLOP-time} for listing only the cost-optimal elements of~$X$ with respect to some linear objective function, for example, minimum weight spanning trees or maximum weight matchings in a graph.
In particular, this includes minimum or maximum cardinality solutions, for example maximum matchings or minimum vertex covers in a graph.

Let $C\seq \mathbb{Z}$, referred as the \defi{cost set}.
Furthermore, let $c \in C^n$ be a cost vector and let~$X_c$ be the elements in~$X$ with minimum cost according to~$c$, i.e., $X_c := \argmin[x \in X \mid c \cdot x]$.
The elements of~$X_c$ lie on a hyperplane in $n$-dimensional space, and so $\conv(X_c)$ is a face of~$\conv(X)$.
In particular, $\conv(X_c)$ is a 0/1-polytope whose edges are edges of~$\conv(X)$.
The problem~LOP on~$X_c$ thus becomes a bi-criteria linear optimization problem with prescription.
In the following we show that by appropriate amplification of the cost vector by a factor of~$n$, we can eliminate the bi-criteria optimization and reduce to a standard~LOP.
Specifically, given $c \in \mathbb{Z}^n$, we define
\begin{equation}
\label{eq:WC}
W(C) := \{-1,0,+1\}+nC=\big\{w+n c\mid w\in\{-1,0,+1\}\text{ and } c\in C\big\},
\end{equation}
and reduce to solving LOPs with weight set~$W(C)$.

The following auxiliary lemma allows us to translate minimization on~$X_c$ to minimization on~$X$ via weight amplification.

\begin{lemma}
\label{lem:weight-ampl}
Let~$X\seq\{0,1\}^n$, $w\in\{-1,0,+1\}^n$, and~$c\in C^n$ with $C\seq \mathbb{Z}$.
Then the weight vector~$w':=w+nc\in W(C)^n$ with $W(C)$ as defined in~\eqref{eq:WC} has the following properties:
\begin{enumerate}[label=(\roman*),leftmargin=8mm, noitemsep, topsep=1pt plus 1pt]
\item For every $y\in X_c$ and $y'\in X\setminus X_c$ we have $w'\cdot y\leq w'\cdot y'$.
Consequently, if $y\in X_c$ and $y'\in X$ satisfy $w'\cdot y>w'\cdot y'$, then we have $y'\in X_c$.
\item Let $\beta\in[n]$ be such that $y_\beta=y'_\beta$ for all $y,y'\in X$.
Then we have $\argmin[y\in X_c\mid w'\cdot y]=\argmin[y\in X\mid w'\cdot y]$.
\end{enumerate}
\end{lemma}

\begin{proof}
To prove~(i), let $y\in X_c$ and $y'\in X\setminus X_c$, i.e., $y$ minimizes the costs and $y'$ does not, in particular $c\cdot y<c\cdot y'$.
As $c$ is an integer vector, we therefore have $c\cdot y\leq c\cdot y'-1$.
Furthermore, as all entries of $w$ are from~$\{-1,0,+1\}$, we have $w\cdot y\leq w\cdot y'+n$ (the same inequality holds with~$y$ and~$y'$ interchanged, but this is not needed here).
Combining these inequalities yields $w'\cdot y=w\cdot y+nc\cdot y\leq w\cdot y'+n+n(c\cdot y'-1)=(w+nc)\cdot y'=w'\cdot y'$, as claimed.

To prove~(ii), note that if all bitstrings from~$X$ agree in the $\beta$th position, then the second inequality from before is strict, which yields the stronger conclusion $w'\cdot y<w'\cdot y'$, which directly proves~(ii).
\end{proof}

Lemmas~\ref{lem:algoA-LOP-c} and~\ref{lem:algoC-LOP-c} below are the analogues of Lemmas~\ref{lem:algoA-LOP} and~\ref{lem:algoC-LOP}, respectively.

\begin{lemma}
\label{lem:algoA-LOP-c}
Let $X \seq \{0,1\}^n$ and $c \in C^n$ with $C\seq \mathbb{Z}$.
Suppose that problem~\upright{LOP} for~$X$ with weight set $W(C)$ can be solved in time~$t_{\upright{LOP}}=\Omega(n)$.
Then problem~\upright{A} for~$X_c$ can be solved in time $\cO(t_{\upright{LOP}}\log n)$.
\end{lemma}

\begin{proof}[Proof of Lemma~\ref{lem:algoA-LOP-c}]
Consider the set~$X_c\seq\{0,1\}^n$, an element~$x\in X_c$ and an interval~$I\in[n]$ as input for problem~A.
We show that
\begin{equation}
\label{eq:mono-c}
\min \{\lambda_I(x,y) \mid y\in X_c-x\} \leq \alpha
\end{equation}
can be decided in time~$\cO(t_{\upright{LOP}})$, from which it follows that $\min \{\lambda_I(x,y)\mid y\in X_c-x\}$ can be computed in time~$\cO(t_{\upright{LOP}}\log n)$ by doing binary search.
We define $P_0$, $P_1$ and $w\in \{-1,0,+1\}^n$ as in~\eqref{eq:PwA}, and we also define $w':=w+nc\in W(C)^n$.
We claim that
\begin{equation}
\label{eq:mu}
\mu:=\min\limits_{y\in X\,\wedge\,y_{P_0}=0\,\wedge\,y_{P_1}=1} w'\cdot y<w'\cdot x=(w+nc)\cdot x=:a
\end{equation}
if and only if~\eqref{eq:mono-c} holds.
Crucially, the minimization in~\eqref{eq:mu} is over the entire set~$X$, whereas the minimization in~\eqref{eq:mono-c} is only over the subset~$X_c\seq X$.

To prove one direction of the claim, note that if $\mu<a$, then there is a $y^*\in X$ with $y^*_{P_0}=0$, $y^*_{P_1}=1$, and $y^*_i\neq x_i$ for some $i\in I$ with $i\leq \alpha$.
It follows that $\lambda_I(x,y^*)\leq \alpha$.
Applying Lemma~\ref{lem:weight-ampl}~(i) shows that $y^*\in X_c$, which implies~\eqref{eq:mono-c}.

To prove the other direction of the claim, if~\eqref{eq:mono-c} holds, then there is $y^*\in X_c-x$ with $\lambda_I(x,y^*)\leq \alpha$, i.e., there is a position~$i\in I$ with $i\leq \alpha$ such that $y^*_i\neq x_i$ and $y^*_j=x_j$ for all $j\geq i+1$, in particular $y^*_{P_0}=0$ and $y^*_{P_1}=1$.
As $y^*_i\neq x_i$ we have $w\cdot y^*<w\cdot x$.
Using that $y\in X_c$ we also have $c\cdot y^*=c\cdot x$.
Combining these observations yields $w'\cdot y^*=w\cdot y^*+nc\cdot y^*<w\cdot x+nc\cdot x=(w+nc)\cdot x=a$ and therefore $\mu<a$.
This completes the proof of the lemma.
\end{proof}

\begin{lemma}
\label{lem:algoC-LOP-c}
Let $X \seq \{0,1\}^n$ and $c\in C^n$ with $C\seq \mathbb{Z}$.
Suppose that problem~\upright{LOP} for~$X$ with weight set~$W(C)$ can be solved in time~$t_{\upright{LOP}}=\Omega(n)$.
Then problem~\upright{C} for~$X_c$ can be solved in time~$\cO(t_{\upright{LOP}})$.
\end{lemma}

\begin{proof}
Consider the set~$X_c\seq\{0,1\}^n$, an element~$x\in X_c$ and an integer~$\beta\in[n]$ with
$N:=\argmin[y\in X_c-x\,\wedge\,\lambda(x,y)=\beta\mid d(x,y)]\neq \emptyset$
as input for problem~C.
We define $P_0$, $P_1$ and $w$ as in~\eqref{eq:PwC}, and we also define $w':=w+nc\in W(C)^n$.

By these definitions we have $d(x,y)=w\cdot(y-x)=(w'-nc)\cdot (y-x)$, and as $x$ is fixed and $c\cdot y$ is the same value for all $y\in X_c$, minimization of $d(x,y)$ is the same as minimization of~$w'\cdot y$.
Consequently, we have
\begin{align*}
N &= \argmin[y\in X_c-x\,\wedge\,\lambda(x,y)=\beta\mid d(x,y)] \\
&=\argmin[y\in X_c\,\wedge\,y_{P_0}=0\,\wedge\,y_{P_1}=1\mid w'\cdot y] \\
&=\argmin[y\in X\,\wedge\,y_{P_0}=0\,\wedge\,y_{P_1}=1\mid w'\cdot y],
\end{align*}
where we used the definitions of~$P_0$ and~$P_1$ in the first step, and Lemma~\ref{lem:weight-ampl}~(ii) in the second step.
\end{proof}

\begin{theorem}
\label{thm:algoPLOP-time-c}
Let $X \seq \{0,1\}^n$ and $c\in C^n$ with $C\seq \mathbb{Z}$.
Suppose that problem~\upright{LOP} for~$X$ with weight set~$W(C)$ as defined in~\eqref{eq:WC} can be solved in time~$t_{\upright{LOP}}=\Omega(n)$.
Then for every tiebreaking rule and every initial vertex~$\tx\in X_c$, Algorithm~\upright{P\sss{}} computes a genlex Hamilton path on the skeleton of~$\conv(X_c)$ starting at~$\tx$ with delay~$\cO(t_{\upright{LOP}}\log n)$.
\end{theorem}

\begin{proof}
Combine Lemmas~\ref{lem:algoA-LOP-c} and~\ref{lem:algoC-LOP-c} with Theorem~\ref{thm:algoP-time}.
\end{proof}

\begin{remark}
\label{rem:polylog}
Suppose that problem~LOP for $X$ with weight set~$W$ can be solved in time~$f(n,M)$, where $M:=\max W$.
Then problem~LOP for~$X$ with weight set~$W(C)$ can be solved in time~$f(n,nM)$.
Often the dependency of~$f$ on~$M$ is \emph{polylogarithmic}, and in those cases~$f(n,nM)$ is bigger than~$f(n,M)$ only by a polylogarithmic factor in~$n$.
Then the delay of Algorithm~P\sss{} is larger than the time for solving the corresponding optimization problem on~$X$ only by a polylogarithmic factor in~$n$.
\end{remark}

\subsection{Eliminating the prescription constraints}

A similar weight amplification trick can be used to reduce linear optimization with prescription to classical \defi{linear optimization} (without prescription).

\begin{enumerate}[label=\mybox{LO},leftmargin=12mm, noitemsep, topsep=1pt plus 1pt]
\item Given a set~$X\seq\{0,1\}^n$ and a weight vector~$w\in W^n$ with $W\seq\mathbb{R}$, compute an element in $N:=\argmin[y\in X \mid w\cdot y]$, or decide that this problem is infeasible, i.e., $N=\emptyset$.
\end{enumerate}

\begin{lemma}
\label{lem:LOP-LO}
Let $X\seq \{0,1\}^n$ and $w\seq W^n$ with $W\seq\mathbb{Z}\cap[-M,+M]$.
Suppose that problem~\upright{LO} with weight set~$W \cup \{-nM,+nM\}$ can be solved in time~$t_{\upright{LO}}=\Omega(n)$.
Then problem~\upright{LOP} with weight set~$W$ can be solved in time~$\cO(t_{\upright{LO}})$.
\end{lemma}

\begin{proof}
Consider $X$ and~$w$ as in the lemma, and sets~$P_0,P_1\seq [n]$ as input for problem~LOP.
We define $Q:=[n]\setminus(P_0\cup P_1)$.
Clearly, we may assume that $|P_0\cup P_1|>0$, or equivalently $|Q|<n$.
We define a weight vector~$w'\in(W\cup\{-nM,+nM\})^n$ by
\[
w_i':=\begin{cases}
     w_i & \text{if } i \in Q, \\
     +nM & \text{if } i \in P_0, \\
     -nM & \text{if } i \in P_1.
     \end{cases}
\]
We claim that
\[N:=\argmin[y\in X\,\wedge\,y_{P_0}=0\,\wedge\,y_{P_1}=1\mid w\cdot y]=\argmin[y\in X\mid w'\cdot y].\]
For every $y\in X$ we define the abbreviations $f(y):=\sum_{i\in Q}w_iy_i$ and $g(y):=\sum_{i\in P_0}y_i-\sum_{i\in P_1}y_i$, so $w'\cdot y=f(y)+nMg(y)$.
Furthermore, we define $X_P:=\{y\in X\mid y_{P_0}=0\,\wedge\,y_{P_1}=1\}$.
Let $y\in X_P$ and $y'\in X\setminus X_P$.
As all entries of $y$ and~$y'$ are from~$\{0,1\}$ and $|w_i|\leq M$ we have $w_i(y_i-y_i')\leq M$ and therefore $f(y)\leq f(y')+|Q|M<f(y')+nM$.
Furthermore, note that $g(y)=-|P_1|$ and $g(y')\geq -|P_1|+1=g(y)+1$.
Combining these inequalities yields
\[w'\cdot y=f(y)+nMg(y)<f(y')+nM+nM(g(y')-1)=f(y')+nMg(y')=w'\cdot y',\]
which proves that $\argmin[y\in X\mid w'\cdot y]=\argmin[y\in X_P\mid w'\cdot y]$.
Now observe that $w'\cdot y=w\cdot y+\sum_{i\in P_0\cup P_1}(w_i'-w_i)y_i$, and the sum on the right hand side of this equation is a constant for all $y\in X_P$, so minimizing $w'\cdot y$ over all~$y\in X_P$ is the same as minimizing~$w\cdot y$ over all~$y\in X_P$.
This proves the claim and thus the lemma.
\end{proof}

Via Lemma~\ref{lem:LOP-LO}, the problem of computing a Hamilton path on the polytope~$\conv(X)$ is reduced entirely to solving linear optimization (problem~LO) over~$X$.
Lemma~\ref{lem:LOP-LO} can be applied in conjunction with Theorem~\ref{thm:algoPLOP-time} or Theorem~\ref{thm:algoPLOP-time-c}; see Table~\ref{tab:fg-bounds}.

\begin{table}
\caption{Delay of Algorithm~P\sss{} obtained from applying Theorems~\ref{thm:algoPLOP-time} and~\ref{thm:algoPLOP-time-c} with Lemma~\ref{lem:LOP-LO}.
We assume that problems~LOP and~LO for~$X$ can be solved in time~$t_{\upright{LOP}}=f(n,M)$ or~$t_{\upright{LO}}=g(n,M)$, respectively, where $M$ is an upper bound on the entries of the weight vector~$w\in W^n$, i.e., $W\seq\mathbb{Z}\cap[-M,+M]$.
We also assume that both of these functions are in~$\Omega(n)$.}
\label{tab:fg-bounds}
\centering
\makebox[0cm]{ 
\begin{tabular}{|p{37mm}|p{25mm}|p{40mm}|p{44mm}|}
\hline
\textbf{Objects} & \textbf{Optimization problem} & \textbf{Weight set} & \textbf{Delay of Algorithm~P\sss{}} \\ \hline
\multirow{2}{37mm}{All elements in $X$} & LOP for $X$ & $\{-1,0,1\}$ & $\cO(f(n,1) \log n)$ \newline Thm.~\ref{thm:algoPLOP-time} \\ \cline{2-4}
                                        & LO for $X$ & $\{-n,-1,0,1,n\}$ & $\cO(g(n,n) \log n)$  \newline Thm.~\ref{thm:algoPLOP-time}+Lemma~\ref{lem:LOP-LO} \\ \hline
\multirow{2}{37mm}{Elements in $X$ of minimum size, i.e., $c=(+1,\ldots,+1)$} & LOP for $X$ & $\{n-1,n,n+1\}$ & $\cO(f(n,n+1) \log n)$ \newline Thm.~\ref{thm:algoPLOP-time-c} \\ \cline{2-4}
                                        & LO for $X$ & $\{-n(n+1),n-1,n,$ \newline ${}\hspace{10mm}n+1,n(n+1)\}$ & $\cO(g(n,n(n+1)) \log n)$ \newline Thm.~\ref{thm:algoPLOP-time-c} + Lemma~\ref{lem:LOP-LO} \\ \hline
\multirow{2}{37mm}{Elements in $X$ of maximum size, i.e., $c=(-1,\ldots,-1)$} & LOP for $X$ & $\{-n-1,-n,-n+1\}$ & $\cO(f(n,n+1) \log n)$\newline Thm.~\ref{thm:algoPLOP-time-c} \\ \cline{2-4}
                                        & LO for $X$ & \small$\{-n(n+1),-n-1,-n,$ \newline ${}\hspace{10mm}-n+1,n(n+1)\}$ & $\cO(g(n,n(n+1)) \log n)$ \newline  Thm.~\ref{thm:algoPLOP-time-c} + Lemma~\ref{lem:LOP-LO} \\ \hline
\multirow{2}{37mm}{$c$-optimal elements in $X$ where $c\in C^n$ and $C\seq\mathbb{Z}\cap[-M,M]$} & LOP for $X$ & $W(C)$  & $\cO(f(n,nM+1) \log n)$ \newline Thm.~\ref{thm:algoPLOP-time-c} \\ \cline{2-4}
                                        & LO for $X$ & $W(C) \cup \{-n(nM+1),$ \newline ${}\hspace{19mm}n(nM+1)\}$ & $\cO(g(n,n(nM+1)) \log n)$ \newline Thm.~\ref{thm:algoPLOP-time-c} + Lemma~\ref{lem:LOP-LO} \\ \hline
\end{tabular}
}
\end{table}

\begin{remark}
The results obtained by applying Lemma~\ref{lem:LOP-LO} are incomparable to those obtained without the lemma (which rely on solving problem~LOP).
Specifically, while eliminating the prescription constraints one arrives at a simpler optimization problem, this comes at the cost of increasing the weights, thus potentially increasing the running time.
Nevertheless, Lemma~\ref{lem:LOP-LO} provides an \emph{automatic} way to exploit every algorithm for the linear optimization problem on~$X$ (problem~LO; without prescription constraints) for computing a Hamilton path on the corresponding 0/1-polytope $\conv(X)$, with provable runtime guarantees for this computation.
\end{remark}

\section{Applications}
\label{sec:appl}

In this section, we show how to apply our general theorems to obtain efficient algorithms for generating Gray codes for different concrete classes of combinatorial objects.
The list of applications shown here is not exhaustive, but exemplary.
More results are derived in short form in Table~\ref{tab:appl}.

\subsection{Vertices of a 0/1-polytope}
\label{sec:appl-poly}

Let $A \in \mathbb{R}^{m \times n}$ and $b \in \mathbb{R}^m$ be such that $P:=\{x \in \mathbb{R}^n \mid Ax \leq b \}$ is a 0/1-polytope.
We let $X\seq \{0,1\}^n$ denote the set of vertices of~$P$, i.e., we have $P=\conv(X)$.
The problem~LOP defined in Section~\ref{sec:reduction} translates to solving the LP
\[
\min \{w\cdot x \mid Ax \leq b \,\wedge\, x_{P_0}=0 \,\wedge\, x_{P_1}=1\}.
\]
By eliminating the prescribed variables from this problem, we see that it is equivalent to the standard LP
\begin{equation}
\label{eq:prob-LP}
\min\{ w \cdot x\mid Ax \leq b\}
\end{equation}
(with modified $A$, $b$, and $w$, but with the same bounds on their sizes).
Thus, invoking Theorem~\ref{thm:algoPLOP-time} we obtain the following result about the vertex enumeration problem for 0/1-polytopes.
Our result improves upon the $\cO(t_{\upright{LP}}\,n)$ delay algorithm of Bussieck and L\"ubbecke~\cite{MR1659922}, and it has the additional feature that the vertices of the polytope are visited in the order of a Hamilton path on the skeleton.

\begin{corollary}
\label{cor:appl-poly}
Let $P=\{x \in \mathbb{R}^n \mid Ax \leq b \}$ be a 0/1-polytope, and suppose that the LP~\eqref{eq:prob-LP} can be solved in time $t_{\upright{LP}}=\Omega(mn)$.
Then for every tiebreaking rule and every initial vertex~$\tx$, Algorithm~\upright{P\sss{}} computes a genlex Hamilton path on the skeleton of~$P$ starting at~$\tx$ with delay~$\cO(t_{\upright{LP}}\log n)$.
\end{corollary}

We can also apply Theorem~\ref{thm:algoPLOP-time-c} to visit only the cost-optimal vertices of~$P$.
For this we also assume that $t_{\upright{LP}}$ depends polylogarithmically on the largest weight
(recall Remark~\ref{rem:polylog}).

\begin{corollary}
\label{cor:appl-poly-c}
Let $P$ and $t_{\upright{LP}}$ be as in Corollary~\ref{cor:appl-poly}, and suppose that $t_{\upright{LP}}$ is polylogarithmic in the largest weight.
Furthermore, let $c \in \mathbb{Z}^n$ and $P_c:=\argmin[x\in P\mid c\cdot x]$.
Then for every tiebreaking rule and every initial vertex~$\tx\in P_c$, Algorithm~\upright{P\sss{}} computes a genlex Hamilton path on the skeleton of~$P_c$ starting at~$\tx$ with delay~$\cO(t_{\upright{LP}}\poly(\log n))$.
\end{corollary}

The initialization time of Algorithm~P\sss{} in both results is the same as the delay, and the required space is the same as the space needed to solve the LP~\eqref{eq:prob-LP}.

\subsection{Spanning trees of a graph}

Let $H$ be a connected $n$-vertex graph with edge set~$[m]$.
We let $X$ denote the set of indicator vectors of spanning trees of~$H$, i.e.,
\[X = \{\indicator_{T} \mid  T\seq[m] \text{ is a spanning tree of } H \}\seq \{0,1\}^m.\]
It is well known (\cite{MR510371}) that the edges of the \defi{spanning tree polytope} $\conv(X)$ are precisely between pairs of trees $T,T'$ that differ in an edge exchange, i.e., there are edges $i,j\in [m]$ such that $T'=T+i-j$.
We thus obtain the following specialization of Algorithm~G for listing all spanning trees of~$H$ by edge exchanges.
The greedy update rule in step~T3 minimizes the larger of the two edges in each exchange.

\begin{algo}{Algorithm~T}{Spanning trees by shortest prefix changes}
Given a connected graph~$H$ with edge set~$[m]$, this algorithm greedily generates all spanning trees of~$H$ by edge exchanges, starting from an initial spanning tree~$\tT$.
\begin{enumerate}[label={\bfseries T\arabic*.}, leftmargin=8mm, noitemsep, topsep=3pt plus 3pt]
\item{} [Initialize] Set $T \gets \tT$.
\item{} [Visit] Visit $T$.
\item{} [Shortest prefix change] Compute the set $N$ of unvisited spanning trees~$T'$ that differ from~$T$ in the exchange of edges~$i,j$ with smallest value $\max\{i,j\}$, i.e., $N\gets \argmin[T'=T+i-j \text{ spanning tree of } H \,\wedge\, T' \text{ unvisited} \mid \max \{i,j\}]$.
Terminate if $N=\emptyset$.
\item{} [Tiebreaker+update~$T$] Pick an arbitrary tree $T'\in N$, set $T\gets T'$ and goto~T2.
\end{enumerate}
\end{algo}

This algorithm for generating spanning trees by edge exchanges has been described before by Merino, M\"utze, and Williams~\cite{MR4473269}.
They gave an implementation that achieves delay~$\cO(m \log n (\log \log n)^3)$.
We now improve on this result using our framework via optimization.

The problem~LOP defined in Section~\ref{sec:reduction} translates to computing a minimum weight spanning tree~$T$ in~$H$ according to some weight function $w\in\mathbb{R}^m$, with the prescription constraints $P_0\cap T=\emptyset$ and $P_1\seq T$, i.e., the edges in~$P_0$ are forbidden, and the edges in~$P_1$ are forced.
This could be achieved by computing the graph~$H'$ that is obtained from~$H$ by deleting the edges in~$P_0$ and contracting the edges in~$P_1$, which may however be costly.
Instead, we solve the problem on the original graph~$H$, but with the modified weight function
\[w_i' := \begin{cases} w_i & \text{if } i \notin P_0 \cup P_1, \\
M & \text{if } i\in P_0, \\
-M & \text{if } i\in P_1,
\end{cases}
\]
where $M$ is chosen so that $M>\max_{i\in[m]} |w_i|$. 
For applying Theorem~\ref{thm:algoPLOP-time} we only need to consider weights~$w\in\{-1,0,1\}^m$, so we can take $M=2$.
For graphs with constantly many distinct edge weights (in our case $\{-2,-1,0,1,2\}$), the minimum spanning tree problem can be solved in time~$\cO(m)$, by a variation of Prim's algorithm that instead of a priority queue uses one list of same weight edges for each possible weight.
Theorem~\ref{thm:algoPLOP-time} thus yields the following corollary.

\begin{corollary}
\label{cor:appl-tree}
Let $H$ be an $n$-vertex graph with edge set~$[m]$.
Then for every tiebreaking rule and every initial spanning tree~$\tT$, Algorithm~\upright{P\sss{}} computes a genlex listing of all spanning trees of~$H$ by edge exchanges starting at~$\tT$ with delay~$\cO(m \log n)$.
\end{corollary}

To generate all cost-optimal spanning trees, we apply Theorem~\ref{thm:algoPLOP-time-c} and combine it with Chazelle's~\cite{MR1866456} algorithm for computing minimum spanning trees, which runs time~$\cO(m\alpha(m,n))$, where $\alpha$ is the functional inverse of the Ackermann function.

\begin{corollary}
\label{cor:appl-tree-c}
Let $H$ be an $n$-vertex graph with edge set~$[m]$, and let $c\in\mathbb{Z}^m$.
Then for every tiebreaking rule and every initial $c$-minimal spanning tree~$\tT$, Algorithm~\upright{P\sss{}} computes a genlex listing of all $c$-minimal spanning trees of~$H$ by edge exchanges starting at~$\tT$ with delay~$\cO(m \alpha(m,n)\log n)$.
\end{corollary}

In both of these results, the initialization time of Algorithm~P\sss{} is the same as the delay, and the required space is the same as for computing a minimum spanning tree.

\subsection{Matchings of a graph}

Let $H$ be an $n$-vertex graph with edge set~$[m]$.
We let $X$ denote the set of indicator vectors of matchings of~$H$, i.e.,
\[X = \{\indicator_M \mid M\seq [m] \text{ is a matching of } H\}\seq \{0,1\}^m .\]
It is well known (\cite{MR371732}) that the edges of the \defi{matching polytope} $\conv(X)$ are precisely between pairs of matchings $M,M'$ that differ in an alternating path or cycle exchange, i.e., there is a path or cycle~$E$ such that $M'=M\triangle E$, where $\triangle$ denotes the symmetric difference.
We thus obtain the following specialization of Algorithm~G for listing all matchings of~$H$ by alternating path/cycle exchanges.
The greedy update rule in step~M3 minimizes the largest of the edges in~$E$.

\begin{algo}{Algorithm~M}{Matchings by shortest prefix changes}
Given a graph~$H$ with edge set~$[m]$, this algorithm greedily generates all matchings of~$H$ by alternating path/cycle exchanges, starting from an initial matching~$\tM$.
\begin{enumerate}[label={\bfseries M\arabic*.}, leftmargin=8mm, noitemsep, topsep=3pt plus 3pt]
\item{} [Initialize] Set $M \gets \tM$.
\item{} [Visit] Visit $M$.
\item
{} [Shortest prefix change] Compute the set $N$ of unvisited matchings~$M'$ that differ from~$M$ in the exchange of an alternating path/cycle~$E$ with smallest value~$\max E$, i.e., $N\gets \argmin[M'=M\triangle E \text{ matching} \,\wedge\, E \text{ path/cycle} \,\wedge\, M' \text{ unvisited} \mid \max E]$.
Terminate if $N=\emptyset$.
\item{} [Tiebreaker+update $M$] Pick an arbitrary matching~$M'\in N$, set $M\gets M'$ and goto~M2.
\end{enumerate}
\end{algo}

The problem~LOP defined in Section~\ref{sec:reduction} translates to computing a minimum weight matching~$M$ in~$H$ according to some weight function $w\in\mathbb{R}^m$, with the prescription constraints $P_0\cap M=\emptyset$ and $P_1\seq M$, i.e., the edges in~$P_0$ are forbidden, and the edges in~$P_1$ are forced.
This can be achieved by computing the graph~$H'$ that is obtained from~$H$ by deleting the edges in~$P_0$ and deleting the vertices that are endpoints of edges in~$P_1$.
We then find a minimum weight matching in the smaller graph~$H'$.
For applying Theorem~\ref{thm:algoPLOP-time} we only need to consider weights~$w\in\{-1,0,1\}^m$, and as LOP is a minimization problem, we can simplify~$H'$ further by deleting all edges with weights~$0$ or~$1$, which yields a graph~$H''$ that has only edges of weight~$-1$.
A minimum weight matching in~$H''$ is therefore a maximum matching in~$H''$, and for finding this we use Micali and Vazirani's~\cite{DBLP:conf/focs/MicaliV80} algorithm, which runs in time~$\cO(m\sqrt{n})$.
Theorem~\ref{thm:algoPLOP-time} thus yields the following corollary.
As Algorithm~P\sss{} minimizes the Hamming distance when moving to the next matching, the alternating path/cycle will in fact always be an alternating path of length~$\leq 3$.

\begin{corollary}
\label{cor:appl-match}
Let $H$ be an $n$-vertex graph without isolated vertices with edge set~$[m]$.
Then for every tiebreaking rule and every initial matching~$\tM$, Algorithm~\upright{P\sss{}} computes a genlex listing of all matchings of~$H$ by alternating path exchanges of length~$\leq 3$ starting at~$\tM$ with delay~$\cO(m\sqrt{n}\log n)$.
\end{corollary}

To generate all cost-optimal matchings w.r.t.\ some cost vector~$c\in\mathbb{Z}^m$, we apply Theorem~\ref{thm:algoPLOP-time-c} and combine it with 
Duan, Pettie and Su's~\cite{MR3763658} algorithm, which runs in time 
$\cO(m\sqrt{n}\log (n|c|))$, or Gabow's~\cite{MR3744699} implementation of Edmond's algorithm, which runs in time~$\cO(mn+n^2\log n)$.
The quantity~$|c|$ is the maximum absolute value of entries of~$c$.
These algorithms maximize the weight instead of minimizing it, but we can simply multiply all weights by~$-1$.

\begin{corollary}
\label{cor:appl-match-c}
Let $H$ be an $n$-vertex graph without isolated vertices with edge set~$[m]$, and let $c\in\mathbb{Z}^m$.
Then for every tiebreaking rule and every initial $c$-minimal matching~$\tM$, Algorithm~\upright{P\sss{}} computes a genlex listing of all $c$-minimal matchings of~$H$ by alternating path/cycle exchanges starting at~$\tM$ with delay
\[\min\Big\{\cO(m\sqrt{n}\log (n|c|)\log n), \cO((mn+n^2\log n)\log n)\Big\}.\]
\end{corollary}

A particularly interesting case is when the cost vector is $c=(-1,\ldots,-1)$, i.e., we obtain all maximum size matchings of~$H$.
In particular, if $H$ has a perfect matching, we can generate all perfect matchings of~$H$.
The weights for our minimization problem will be $\{-m-1,-m,-m+1\}$, and for the corresponding maximization problem they will be~$\{m-1,m,m+1\}$.

\begin{corollary}
\label{cor:appl-match-max}
Let $H$ be an $n$-vertex graph without isolated vertices with edge set~$[m]$.
Then for every tiebreaking rule and every initial maximum matching~$\tM$, Algorithm~\upright{P\sss{}} computes a genlex listing of all maximum matchings of~$H$ by alternating path/cycle exchanges starting at~$\tM$ with delay $\cO(m\sqrt{n}(\log n)^2)$.
\end{corollary}

In all three of these results, the initialization time of Algorithm~P\sss{} is the same as the delay, and the required space is the same as for computing a maximum (weight) matching.

\section{Duality between Algorithm~P and Algorithm~J}
\label{sec:duality}

There is an interesting duality between the generation framework proposed here and the permutation language framework due to Hartung, Hoang, M\"utze and Williams~\cite{MR4391718}.
Their framework encodes combinatorial objects by permutations of length~$n$, and the local change operation to go from one permutation to the next is a cyclic substring shift by one position, subject to the constraint that the largest value~$j$ in the substring wraps around and moves to the other end of the substring.
This operation is referred to as a~\defi{jump of~$j$}, the number of \defi{steps} of the jump is one less than the length of the shifted substring, and its \defi{direction} is the direction of movement of~$j$ in the substring.
For example $24135\rightarrow 21345$ is a jump of the value~4 by 2~steps to the right.
Similarly, $123\cdots n\rightarrow n123\cdots (n-1)$ is a jump of the value~$n$ by $n-1$ steps to the left.
Their framework is based on a simple greedy algorithm, \defi{Algorithm~J}, which attempts to generate a set of permutations by jumps, by repeatedly and greedily performing a shortest possible jump of the largest possible value so that a previously unvisited permutation is created.
Compare this to our Algorithm~P, which repeatedly and greedily performs a prefix change of shortest possible length so that a previously unvisited bitstring is created.
In fact, these two algorithms are dual to each other:
While Algorithm~J works based on values, Algorithm~P works based on positions.

Specifically, we can simulate Algorithm~P by Algorithm~J, by encoding bitstrings of length~$n$ by permutations of length~$n+1$.
This encoding is defined inductively using the suffix tree generated by a run of Algorithm~P, and can be done so that a prefix change of length~$d$ on the bitstrings corresponds to a jump of the value~$j:=n-d+2$ by $j-1$ steps in the permutations (i.e., $j$ jumps across all values smaller than itself).

Conversely, Algorithm~J can also be simulated by Algorithm~P (even though Algorithm~P knows only~0s and~1s).
The idea is to encode each permutation~$\pi=a_1\cdots a_{n+1}$ by its \defi{inversion table} $c=c_1\cdots c_{n+1}$, where $0\leq c_i\leq i-1$ counts the number of entries of~$\pi$ that are smaller than~$i$ and to the right of~$i$.
The value~$c_i$ of the inversion table is encoded by a bitstring of length~$i$ with a single~1 at position~$c_i$ (counted from~0), and the entire inversion table is encoded by concatenating these bitstrings in reverse order.
In the simulation, a jump of the value~$i$ by $d$ steps in the permutation changes exactly the entry~$c_i$ by $\pm d$, and so this 1-bit moves by~$d$ positions in the corresponding encoding.
By the concatenation in reverse order, a jump of the largest possible value corresponds to a change of the shortest possible prefix.

\section{Open questions}
\label{sec:open}

We conclude this paper with some open questions.

\begin{itemize}[itemsep=0ex,parsep=0ex,leftmargin=2.5ex]
\item
Does our theory generalize to nonbinary alphabets (cf.~ the discussion in the previous section)?
More specifically, instead of encoding a class of objects~$X$ as bitstrings~$X\seq\{0,1\}^n$, we may use $X\seq\{0,1,\ldots,b-1\}$ for some integer~$b\geq 2$.
Many of the concepts introduced here generalize straightforwardly, including Algorithm~G, genlex ordering, suffix trees, and Theorem~\ref{thm:algoG-genlex} (recall in particular Remark~\ref{rem:greedy}).
However, it is not clear how to generalize prefix graphs so that they arise as the skeleta of the corresponding polytopes.
Most importantly, not all polytopes admit a Hamilton path, unlike 0/1-polytopes via Naddef and Pulleyblank's~\cite{MR762893} result.

\item
Is there an improved amortized analysis of Algorithm~P, maybe for certain classes of combinatorial objects?
Such an analysis would need to consider the heights at which branchings occur in the suffix tree, and this may depend on the ordering of the ground set.
For combinations (=bases of the uniform matroid) such an improved amortized analysis can be carried out and provides better average delay guarantees (see~\cite{MR4473269}).

\item
Can we get Gray code listings for particular classes of objects from our framework that have interesting additional structural properties?
For example, Knuth~\cite{MR3444818} asked whether there is a simple ordering of the spanning trees of the complete graph~$K_n$ by edge exchanges.
It is not completely well-defined what `simple' means, but certainly efficient ranking and unranking algorithms would be desirable.
\end{itemize}

\bibliographystyle{alpha}
\bibliography{refs}

\end{document}